\documentclass{article}

\usepackage[T1]{fontenc}
\usepackage[outline]{contour}
\usepackage{todonotes}
\usepackage[numbers]{natbib}
\usepackage{authblk}
\usepackage{fullpage}  
\usepackage{amsmath,amssymb,mathrsfs}
\usepackage{amsthm}
\usepackage{color}
\usepackage[all,color]{xy}
\usepackage{graphicx}
\usepackage{comment}
\usepackage{tikz}
\usepackage{tikzscale}
\usepackage{subfigure}
\usetikzlibrary{calc}
\usepackage{algorithmic,algorithm2e}
\usepackage{multirow}
\usepackage{gensymb}
\usepackage{filecontents}
\usepackage{todonotes}
\usepackage{titlesec}
\usepackage{pgfgantt}
\usepackage{rotating}
\usepackage{placeins}
\usepackage{enumitem}  
\usepackage{IEEEtrantools}
\usetikzlibrary{calc,intersections,arrows,shapes.geometric,decorations.pathmorphing} \usepackage{xcolor} 
\usepackage{hyperref}
\usepackage{cleveref}

\tikzset{%
    add/.style args={#1 and #2}{
        to path={%
 ($(\tikztostart)!-#1!(\tikztotarget)$)--($(\tikztotarget)!-#2!(\tikztostart)$)%
  \tikztonodes},add/.default={.2 and .2}}
}  
  
\definecolor{grn}{RGB}{57,203,209}
\definecolor{bl}{RGB}{72,61,139}

\newcounter{iitem}

\newenvironment{iitem}{
  \refstepcounter{iitem}

  \enumerate
}{
  \endenumerate
}

\newcommand{\sbf}[1]{\fontseries{sb} {#1}}
  
\newtheorem{definition}{Definition}
\newtheorem{lemma}{Lemma}
\newtheorem{theorem}{Theorem}
\newtheorem{claim}{Claim}
\newtheorem{corollary}{Corollary}
\newtheorem{observation}{Observation}

\newtheorem{conjecture}{Conjecture}

\newtheorem{question}{Question}

\title{On problems related to crossing families}

\author[ ]{William Evans and Noushin Saeedi}
\affil[ ]{\small{Department of Computer Science, University of British Columbia, Canada}}
\affil[ ]{\textit {\{will,~noushins\}@cs.ubc.ca}}

\date{}

\titleformat{\paragraph}
{\normalfont\normalsize\bfseries}{\theparagraph}{1em}{}
\titlespacing*{\paragraph}
{0pt}{3.25ex plus 1ex minus .2ex}{1.5ex plus .2ex}

\begin{document}

\maketitle
\begin{abstract}
Given a set of points in the plane, a \emph{crossing family} is a collection of segments, each joining two of the points, such that every two segments intersect internally.
Aronov et al. [Combinatorica,~14(2):127-134,~1994]
proved that any set of $n$ points contains a crossing family of size $\Omega(\sqrt{n})$. They also mentioned that there exist point sets whose maximum crossing family uses at most $\frac{n}{2}$ of the points. We improve the upper bound on the size of crossing families to $5\lceil \frac{n}{24} \rceil$. We also introduce a few generalizations of crossing families, and give several lower and upper bounds on our generalized notions.

\end{abstract}
\section{Introduction}
Let $P$ be a set of $n$ points in general position in the plane.
A collection of line segments, each joining two of the points, is called a \emph{crossing family} if every two segments intersect internally. Let $\text{crf}(P)$ denote the size of the maximum crossing family in $P$, and let $\text{crf}(n)=\min_{|P|=n} \text{crf}(P)$, where the minimum is taken over all $n$-point sets $P$ in general position in the plane. \citet{Pach-crf} studied the size of $\text{crf}(n)$. They noted that a set of $n$ points chosen at random in a unit disc, ``almost surely'' has a linear-sized crossing family, and that there are point sets whose maximum crossing family uses at most $\frac{n}{2}$ of the points. They proved that any set of $n$ points contains a crossing family of size at least $\Omega(\sqrt{n})$. It is conjectured that $\text{crf}(P)=\Theta(n)$.

We  improve the upper bound on the size of crossing families. We consider more general variants of ``crossings'' and restricted classes of point sets. We study some generalized notions of crossing families, where we consider both combinatorial and geometric generalizations.

Point sets $A$ and $B$ are \emph{separable} if they can be separated by a line. A point set $A$ \emph{separates} point set $B$ from $C$ if $A$ and $B \cup C$ are separable and every line through two points in $A$ has all of $B$ on one side and all of $C$ on the other side.
We show that any crossing family of a point set $A \cup B \cup C$ such that $A$ separates $B$ from $C$ has all its segments incident to one set (i.e. $A$ or $B$ or $C$). We exploit this ``separating property'' between subsets of points and design a template for constructing $n$-point sets whose crossing family is of size at most $5\lceil \frac{n}{24}\rceil$ (see Section~\ref{sec:up}).

A point set $A$ \emph{avoids} a point set $B$ if no line formed by a pair of points in $A$ intersects the convex hull of $B$. $A$ and $B$ are \emph{mutually avoiding} if $A$ avoids $B$ and $B$ avoids $A$. A point set that can be partitioned into separable point sets $A$ and $B$ such that $A$ avoids $B$ is \emph{$1$-avoiding}.
We study the combinatorial properties of crossing families in $1$-avoiding point sets. We relax the combinatorial properties of crossing families and introduce the notion of ``side compatible subsets''. As a step towards finding maximum crossing families in $1$-avoiding point sets, we give linear bounds on side compatible subsets for some simplified versions of the problem (see Section~\ref{sec:sc}).

We also study some geometric generalizations of crossing families.
A \emph{spoke set} for $P$ is a set $\mathcal{L}$ of non-parallel lines such that each open unbounded region in the arrangement $\mathcal{L}$ has at least one point of $P$. The size of a spoke set $\mathcal{L}$ is the number of lines in $\mathcal{L}$.
If $\text{crf}(P)=k$, then the size of the largest spoke set for $P$ is at least $k$. 
This is because by rotating the supporting line of each segment in a crossing family clockwise infinitesimally about the midpoint of the segment, we obtain a spoke set. \citet{gcrf17} studied the properties of spoke sets in the dual plane. He claimed any $1$-avoiding point set of size $n$ has a spoke set of size $\frac{n}{4}$, but as we explain in Section~\ref{sec:ss}, the proof does not seem to be correct.

We introduce a generalized notion for the dual of a spoke set, called an ``M-semialternating path'' (see Section~\ref{sec:m}). M-semialternating paths are related to pseudolines in two-coloured line arrangements that intersect (in order) lines of alternate colours in the arrangement.
We show that the sizes of
different types of
M-semialternating paths are connected and give upper bounds on the size of certain M-semialternating paths. We also improve the upper bound on the size of spoke sets from $\frac{9n}{20}$ to $\frac{n}{4}+1$. Lastly, given a point set, we study the family of segments (each joining two of the points) such that for every pair, the supporting line of one intersects the interior of the other. We call such a family of segments a \emph{stabbing family}. We show that the size of the largest stabbing family in an $n$-point set is $\frac{n}{2}$ (see Section~\ref{sec:stab}).

During the final preparation of this manuscript, we became aware of a recent concurrent work by \citet{pach-new}, showing that any set of $n$ points in general position has a crossing family of size at least $\dfrac{n}{2^{O(\sqrt{\log{n}})}}$. The methods used in our work and theirs are different, and even though we study generalized notions of crossing families or restricted classes of point sets, our results
are different from theirs and not implied by their work.

\section{Related work}
\subsection{Crossing families}
\citet{Pach-crf} introduced the notion of crossing families
and studied the size of $\text{crf}(n)$. They noted that for non-convex point sets, the maximum crossing family may contain at most $\frac{n}{2}$ points.
Two equal-sized disjoint sets $A$ and $B$ can be \emph{crossed} if there exists a crossing family exhausting $A$ and $B$ such that  each line segment connects a point in $A$ to a point in $B$. \citet{Pach-crf} characterized the properties of pairs of point sets that can be crossed. Consider point sets $A$ and $B$ which are separated by a line.  A point $a$ in $A$ sees a point $b$ in $B$ at rank $i$ if $b$ is the $i$-th point seen counterclockwise  from $a$ (starting from the direction of a separating line). Two sets $A$ and $B$ with cardinality $s$ obey the \emph{rank condition} if there exist labelings $a_1,a_2,\dots,a_s$ and $b_1, b_2, \dots, b_s$ such that for all $i$, $a_i$ sees $b_i$ at rank $i$ and vice versa. If the labelings are such that for all $i$ and $j$, $a_i$ sees $b_j$ at rank $j$ and vice versa, then the sets obey the \emph{strong rank condition}.
An $n$-point set is \emph{dense} if the ratio of the maximum distance between any pair of points to the minimum distance (between any pair of points) is $O(\sqrt{n})$.

\citet{Pach-crf} proved that two sets $A$ and $B$ can be crossed if and only if they obey the rank condition. $A$ and $B$ are mutually avoiding if and only if they obey the strong rank condition. They showed that any set of $n$ points contains a pair of mutually avoiding subsets of size $\Theta{\sqrt{n}}$,
and hence $\text{crf}(n)=\Omega(\sqrt{n})$.
 \citet{Valtr} constructed a dense $n$-point set that contains no pair of mutually avoiding subsets of size more than $O(\sqrt{n})$.
Since the strong rank condition is much stronger than the condition needed for having a crossing family (i.e. rank condition), it is believed that the lower bound can be improved. In fact, the conjecture is that $\text{crf}(n)=\Theta(n)$.

\citet{Pach-crf} noted that a set of $n$ points chosen at random in a unit disc, ``almost surely'' has a linear-sized crossing family. \citet{Valtr-dense} showed that the maximum crossing family in a dense point set is almost linear.
\citet{sol-hl} proved that given a set $P$ of $2n$ points in general position, $\text{crf}(P)=n$ if and only if $P$ has exactly $n$ halving lines. A \emph{halving line} in a point set $P$ with an even number of points, is a line through two points of $P$ such that both its half-planes contain exactly the same number of points.

\subsection{Generalizations of crossing families}

Several generalizations of the notion of crossing families have been studied.

In recent years, there has been growing interest in studying pairwise crossings among other objects (rather than segments) formed by a given point set. A \emph{geometric graph} is a graph drawn in the plane so that the vertices are represented by points in general position and the edges are represented by straight line segments connecting the corresponding points. A geometric graph is \emph{complete} if there is an edge between every pair of points. An \emph{H-crossing family} in a complete geometric graph is a set of vertex disjoint isomorphic copies of $H$ that are pairwise crossing, where two geometric subgraphs $h$ and $h'$ \emph{cross} if there is an edge is $h$ that crosses an edge of $h'$. A crossing family can be defined as a $K_2$-crossing family (in a complete geometric graph).
~\citet{tverberg} showed that every complete geometric graph contains a $K_3$-crossing family of size $\lfloor \frac{n}{3}\rfloor$ (which is tight). \citet{cr-ham} showed the tight bound of $\frac{n}{4}$ for $P_4$-crossing families. \citet{ghc18}
showed that every complete geometric graph contains a $P_3$-crossing family of size $O(\sqrt{\frac{n}{2}})$, a $K_{1,3}$-crossing family of size $\frac{n}{6}$, and a $K_4$-crossing family of size $\frac{n}{4}$.

Crossing families have also been studied in \emph{topological graphs} where the edges are represented by Jordan curves (rather than straight line segments). Note that if two edges share an interior point, they must properly cross at that point. A topological graph is \emph{simple} if every pair of edges intersect at most once (i.e. at a common endpoint or at a proper crossing). A graph is \emph{$k$-quasi-planar} if it can be drawn as a topological graph with no $k$ pairwise crossing edges. It is conjectured that for any fixed $k\ge 2$, there exists $c_k$ such that every $k$-quasi-planar graph on $n$ vertices has at most $c_kn$ edges, where $c_k$ is a constant that depends on $k$. \citet{quasi-planar} proved the conjecture for $k=3$. \citet{quasi-4} proved the conjecture for $k=4$. For $k>4$, the best known upper bound on the maximum number of edges in $k$-quasi-planar graphs in which no pair of edges intersect in more than $t$ points is $2^{\alpha(n)^{c}}n \log{n}$, where $\alpha(n)$ denotes the inverse Ackermann function and $c$ depends only on $k$ and $t$~\cite{suk-k-quasi}. For simple $k$-quasi-planar graphs, the best known upper bound is $c_k n \log(n)$ \cite{suk-k-quasi}. \citet{valtr-kp} showed that for any fixed $k \ge 3$, any geometric graph on $n$ vertices with no $k$ pairwise parallel edges contains at most $c_k n$ edges, where two edges are \emph{parallel} if the intersection point of their supporting lines is not on either of the edges. He proved that any geometric graph on $n$ vertices with no $k$ pairwise crossing edges contains at most $c_k n \log(n)$ edges. He also showed that the same bound holds when the edges are drawn as $x$-monotone curves~\cite{valtr-kc}. 
~\citet{k10} gave a simple $3$-quasi-planar drawing of $K_{10}$ (i.e. no three edges are pairwise crossing). It is known that the maximum crossing family of any geometric graph on ten vertices is at least three~\cite{g-k10}. A (topological) graph is \emph{$k$-planar} if it can be drawn in the plane such that no edge is crossed more than $k$ times. \citet{k-quasi-rel} proved that for $k \ge 3$, every simple topological $k$-planar graph can be redrawn (by rerouting some edges) to become $(k+1)$-quasi-planar.

\citet{gcrf17} studied another generalization of crossing families called spoke sets (which was first introduced by \citet{bose}).
Any crossing family of size $k$ guarantees a spoke set of size $k$, but the reverse is not true.
~\citet{gcrf17} characterized the family of ``spoke matchings'', which are geometric matchings in $2k$-point sets admitting a spoke set $\mathcal{L}$ of size $k$, where each matching edge connects the points in antipodal regions of the arrangement $\mathcal{L}$.
He also studied the properties of spoke sets in the dual plane.
He claimed that any $1$-avoiding $n$-point set
has a spoke set of size $\frac{n}{4}$; however, the proof does not seem to be correct. Lastly, he proved that there exist $n$-point sets whose largest spoke set is of size at most $\frac{9}{20}n$ \cite{sch,gcrf17}.

A \emph{cell-path} of a line arrangement is a sequence of cells in the arrangement such that any two consecutive cells of the sequence share a boundary edge and no cell appears more than once. A cell-path is \emph{alternating} if the common edges of consecutive cells alternate in color.
The dual of a spoke set corresponds to a cell-path with certain properties.
\citet{cell-path-Linda} and \citet{cell-path-Aich} studied the existence of long cell-paths in line arrangements. They also studied the bicolored version of the problem, in which they look for long alternating cell-paths. These problems have close connections to problems on point sets in the plane (through duality). \citet{rb-points} give a survey on combinatorial problems on bicolored points in the plane.

\section{Upper bound on crossing families}\label{sec:up}
In this section, we show that there exist $n$-point sets whose crossing family is of size at most $5\lceil\frac{n}{24}\rceil$.
\begin{definition}[separating property]
Let $A$, $B$, and $C$ be three disjoint point sets. We say $A$ \emph{separates} $B$ from $C$ if
\begin{enumerate}[label={(\bfseries C\arabic*)}]
\item $A$ and $B \cup C$ are separable by a line\label{sep-c1}, and
\item every line through two points in $A$ has all of $B$ on one side and all of $C$ on the other.\label{sep-c2}
\end{enumerate}
\label{def:separate}
\end{definition}

Let $P_1$, $P_2$, and $P_3$ be three disjoint point sets. Let $i \in \{1,2,3\}$. The three sets have the \emph{separating property} if, for some $i$, $P_i$ separates $P_{i-1}$ from $P_{i+1}$, where indices are arithmetic modulo $3$. If for all $i$, $P_i$ separates $P_{i-1}$ from $P_{i+1}$, then the three sets satisfy the \emph{full separating} property.
Let $L(P_i)$ denote the set of all lines through two points in $P_i$.
For three point sets $P_1$, $P_2$, and $P_3$ with the full separating property, we define the \emph{core} to be the intersection of all bounded regions formed by any three lines $l_1$, $l_2$, and $l_3$, where $l_1 \in L(P_1)$, $l_2 \in L(P_2)$, and $l_3 \in L(P_3)$. We refer to a segment connecting a point in $P_i$ to a point in $P_j$ as a \emph{$P_iP_j$-segment} or a segment of \emph{type} $P_iP_j$.
A segment is \emph{incident} to a point set $P$ if it has an endpoint in $P$.
We say a set $\mathcal{S}$ of segments \emph{emanate} from a point set $P$ if all segments in $\mathcal{S}$ have an endpoint in $P$.

\begin{lemma}
Let $A$, $B$, and $C$ be three disjoint point sets such that $A$ separates $B$ from $C$. Any crossing family in $A \cup B \cup C$ emanates from $A$ or $B$ or $C$.
\label{lem:split}
\end{lemma}

\begin{proof}
Since $A$ separates $B$ from $C$, no $AB$-segment crosses an $AC$-segment; otherwise, the line through the the endpoints of these segments in $A$ would have at least one point from $B$ and one point from $C$ on the same side, which violates \ref{sep-c2}. Thus, no crossing family can contain an $AB$-segment and an $AC$-segment.

Let $\mathtt{CH}(S)$ denote the convex hull of $S$, where $S$ is a point set. Since $A$ separates $B$ from $C$, by Definition~\ref{def:separate}, $\mathtt{CH}(B \cup C) \cap \mathtt{CH}(A)=\emptyset$ (otherwise, $A$ and $B\cup C$ are not separable). Moreover, $\mathtt{CH}(A \cup B) \cap \mathtt{CH}(C)=\emptyset$ and $\mathtt{CH}(A \cup C) \cap \mathtt{CH}(B)=\emptyset$ (otherwise, property~\ref{sep-c2} is violated). This implies that no crossing family contains an $XY$-segment and a $ZZ$-segment, where $X,Y,Z \in \{A,B,C\}$ and $X,Y \neq Z$. Therefore, any crossing family emanates from $A$ or $B$ or $C$.
\end{proof}

Lemma~\ref{lem:split} immediately implies that there exist sets of $n$ points whose crossing family is of size at most $\lceil\frac{n}{3}\rceil$.
We extend the idea of having the separating property among certain subsets of the point set to obtain a better upper bound. First, we describe a template for constructing point sets with maximum crossing family of size at most $\lceil\frac{n}{4}\rceil$. We then modify the construction to improve the upper bound to $5\lceil\frac{n}{24}\rceil$.  

We start with a set of four points in non-convex position. We denote the points on the convex hull by $p_1$, $p_2$, and $p_3$, and the point inside by $q$. Note that there is no crossing between the segments joining any two of these points. Now we grow a small disc around each point. Call the discs $D_1$, $D_2$, $D_3$, and $D_q$ accordingly.  We replace each point $p$ with $\frac{n}{4}$ points such that
\begin{enumerate}[label=(\arabic*)]
\item the new points are inside the disc around $p$,
\item for all $i \in [3]$, the new points in $D_i$ are almost on a line, and
\item for all $i \in [3]$, $D_i$ separates $D_{i+1}$ from $D_{i-1} \cup D_q$. The indices are arithmetic modulo $3$.
\end{enumerate}

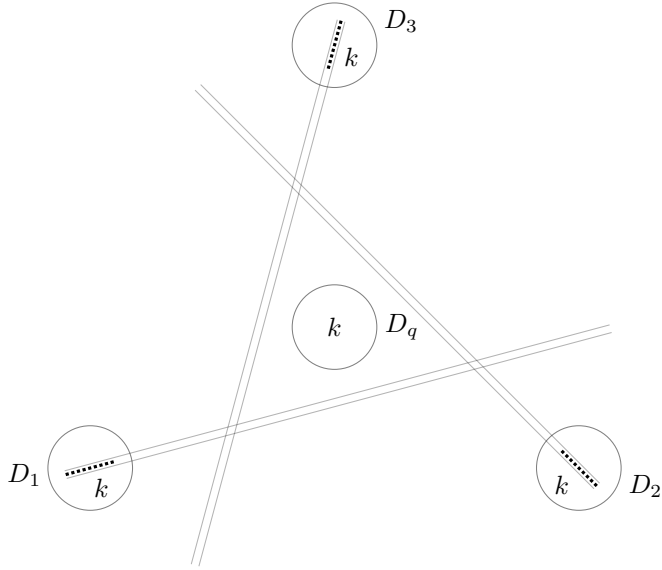
\begin{figure}[h]
\begin{center}
\begin{tikzpicture} [scale=.75]
	\coordinate (D) at (0,0);
	\path (D) -- +(90: 5 cm) coordinate (C);
	\path (D) -- +(330: 5 cm) coordinate (B);
	\path (D) -- +(210:5 cm) coordinate (A);
	
	\path (A) -- +(.2,-.35) coordinate (N1);
	\node at (N1) {{$k$}};	
	\path (B) -- +(-.3,-.3) coordinate (N2);
	\node at (N2) {{$k$}};	
	\path (C) -- +(.3,-.2) coordinate (N3);
	\node at (N3) {{$k$}};
	\node at (D) {{$k$}};

	\foreach \point in {A,B,C,D}
		\draw[opacity=.5] (\point) circle (0.75cm);

	\path let \p1=(A) in (\x1,\y1) -- +(195:.45cm) coordinate (a);
	\path let \p1=(a) in (\x1,\y1) -- +(105:.07cm) coordinate (a1);
	\path let \p1=(a) in (\x1,\y1) -- +(285:.07cm) coordinate (a2);
	\draw[very thick, densely dotted] (a) -- +(15:.9cm);
	\draw[opacity=0.3] (a1) --  +(15:10cm);
	\draw[opacity=0.3] (a2) -- +(15:10cm);

	\path let \p1=(B) in (\x1,\y1) -- +(315:.45cm) coordinate (b);
	\path let \p1=(b) in (\x1,\y1) -- +(45:.07cm) coordinate (b1);
	\path let \p1=(b) in (\x1,\y1) -- +(225:.07cm) coordinate (b2);
	\draw[very thick, densely dotted] (b) -- +(135:.9cm);
	\draw[opacity=0.3] (b1) --  +(135:10cm);
	\draw[opacity=0.3] (b2) -- +(135:10cm);

	\path let \p1=(C) in (\x1,\y1) -- +(75:.45cm) coordinate (c);
	\path let \p1=(c) in (\x1,\y1) -- +(165:.07cm) coordinate (c1);
	\path let \p1=(c) in (\x1,\y1) -- +(345:.07cm) coordinate (c2);
	\draw[very thick, densely dotted] (c) -- +(255:.9cm);
	\draw[opacity=0.3] (c1) --  +(255:10cm);
	\draw[opacity=0.3] (c2) -- +(255:10cm);
	
	\node [left] at (a) {{$D_1\hspace{.2cm}$}};
	\node [right] at (b) {{\hspace{.3cm}$D_2$ }};
	\node [right] at (c) {{\hspace{.2cm}$\hspace{.25cm}D_3$ }};
	\path let \p1=(D) in (\x1,\y1) -- +(.6,0) coordinate (d);
	\node [right] at (d) {{$\hspace{.1cm}D_q$ }};		
\end{tikzpicture}
\end{center}
\caption{A set of $4k$ points with maximum crossing family of size at most $k$.}
\label{fig:nover4}
\end{figure}

See Figure~\ref{fig:nover4}.
The following lemma implies that the segments forming a crossing family in this configuration, should all have an endpoint in the same disk (thus, the maximum crossing family of such a point configuration is of size at most $\frac{n}{4}$).

\begin{lemma}
Let $A,B,C$, and $D$ be four point sets
such that the set $\{A, B, C\}$ satisfies full separating property and $D$ is inside the corresponding core. Any crossing family in $A \cup B \cup C \cup D$ emanates from one of $A,B,C,$ or $D$.
\label{lem:nover4}
\end{lemma}

\begin{proof}
Let $P=A \cup B \cup C \cup D$. Note that for any $X,Y \in \{A,B,C,D\}$, $\mathtt{CH}(X \cup Y)$ and $\mathtt{CH}(P \setminus (X \cup Y))$ are disjoint. Thus, no crossing family can contain an $XY$-segment and a $WZ$-segment if $\{X,Y\} \cap \{W,Z\} =\emptyset$, where $X,Y,Z,W \in \{A,B,C,D\}$.

Let $\mathcal{D}=\{A,B,C,D\}$. Note that (since $A,B$, and $C$ satisfy the full separating property and $D$ is inside the corresponding core) every triple of sets in $\mathcal{D}$ satisfies the separating property. Thus, for any triple of sets, the segments of a crossing family induced by them have an endpoint in the same set. As a result, a crossing family not emanating from a set $X\in \mathcal{D}$ should be incident to all four sets. However, any collection of segments that is incident to all four sets without emanating from one of them must contain an $XY$-segment and a $WZ$-segment, where $\{X,Y\} \cap \{W,Z\}=\emptyset$, and hence cannot form a crossing family.
\end{proof}

Note that in the point configuration illustrated in Figure~\ref{fig:nover4}, there is no restriction on the orientation of the points in $D_q$, and all disks contain the same number of points.
We modify this configuration such that a more specific positioning of the points in $D_q$ allows us to place more points in $D_q$ without increasing the size of the maximum crossing family.

In the following, we describe how to construct such a point configuration in more detail. See Figure~\ref{fig:5nover24}.
We start by putting three discs $D_1$, $D_2$, and $D_3$, each containing $5k$ almost collinear points, such that they satisfy the full separating property. In the core defined by them, we put six more discs $\{A_i,S_i \mid i \in [3]\}$, where each $S_i$ is a ``super'' disk containing two disks $B_i$ and $C_i$. Each disk $A_i, B_i$, and $C_i$ contains $k$ points.
An \emph{$X$-disk}, where $X \in \{A,B,C,D,S\}$ refers to any of the disks $X_1,X_2$, or $X_3$.
Each disk $X_i$ (where $X \in \{A,B,C,D,S\}$) is obtained by rotating $X_{i-1}$ counterclockwise $\frac{2\pi}{3}$ around the origin (marked with $\times$ in Figure~\ref{fig:5nover24}). The indices are arithmetic modulo $3$. Let $P$ denote the set of all points. $P$ is \emph{$3$-fold symmetric} with respect to the origin, that is, it can be partitioned into three sets (called \emph{wings}) of size $\frac{|P|}{3}$ such that each wing rotated by $\frac{2\pi}{3}$ and $\frac{4\pi}{3}$ (around the origin) gives the other two wings.
Each wing is composed of $A_i,B_i,C_i,D_i$ (for some $i \in [3]$) and the order of the disks along each wing in increasing distance from the origin is $A_i,B_i,C_i,D_i$.
For each disk $X$, let $c(X)$ denote the center of $X$. If a disk $X$ contains a point $p$, we write $p \in X$. (For simplicity in writing, we may treat a disk $X$ as the set of points inside disk $X$ at times.)
Let $\mathcal{D} =\{A_i,B_i,C_i,D_i \mid i \in [3]\}$.
Every pair of disks in $\mathcal{D}$ are disjoint, and each disk is infinitesimally small so that for every triple $\langle X,Y,Z \rangle$ of disks (where $X,Y,Z \in \mathcal{D}$), the orientation of $\langle x,y,z \rangle$, where $x \in X$, $y\in Y$, $z \in Z$, is the same as the orientation of $\langle c(X),c(Y),c(Z) \rangle$. (Note that $S_i$ contains $B_i \cup C_i$, hence $S_i \notin \mathcal{D}$.)
A triple $\langle X,Y,Z \rangle$ of disks, where $X,Y,Z \in \mathcal{D}$, has a \emph{positive orientation} if the circle through points $\langle c(X),c(Y),c(Z) \rangle$ is traversed counterclockwise when we encounter the points in cyclic order $ c(X),c(Y),c(Z),c(X)$. A triple of disks in $\mathcal{D}$ has a \emph{negative orientation} if the circle through $\langle c(X),c(Y),c(Z)\rangle$ is traversed clockwise.
The set of disks $\{A_i,B_i,C_i \mid i \in [3]\}$ and the points inside them are such that
\begin{itemize}
\item the set $\{c(D_i),c(A_i) \mid i \in[3] \}$ has a crossing family of size three (the set of segments $\{\overline{c(D_i)c(A_{i+1})} \mid i \in[3]\}$ forms a crossing family), 
\item the maximum crossing family of $\{c(D_i),c(B_i) \mid i \in[3] \}$ is of size two,
\item the maximum crossing family of $\{c(D_i),c(C_i) \mid i \in[3] \}$ is of size two,
\item $\langle D_i,A_{i+1},B_i \rangle$ and $\langle D_i,A_{i+1},C_i \rangle$ have negative and positive orientations, respectively,
\item $\langle B_i,A_{i+1},A_i \rangle$ and $\langle C_i,A_{i+1},A_{i+2} \rangle$ have negative and positive orientations, respectively,
\item $\langle D_i,A_i,B_i \rangle$ and $\langle D_i,A_{i+2},C_i \rangle$ have positive and negative orientations, respectively,
\item $\langle C_i,B_{i+1},A_i \rangle$ has a positive orientation, and
\item the points in $S_i$ are almost collinear, and every line through two points in $S_i$ separates $D_i$ from $P \setminus (D_i \cup S_i)$.
\end{itemize}
\begin{figure}[h!]
\begin{center}
\hspace*{-.8cm}  
\includegraphics[scale=.65]{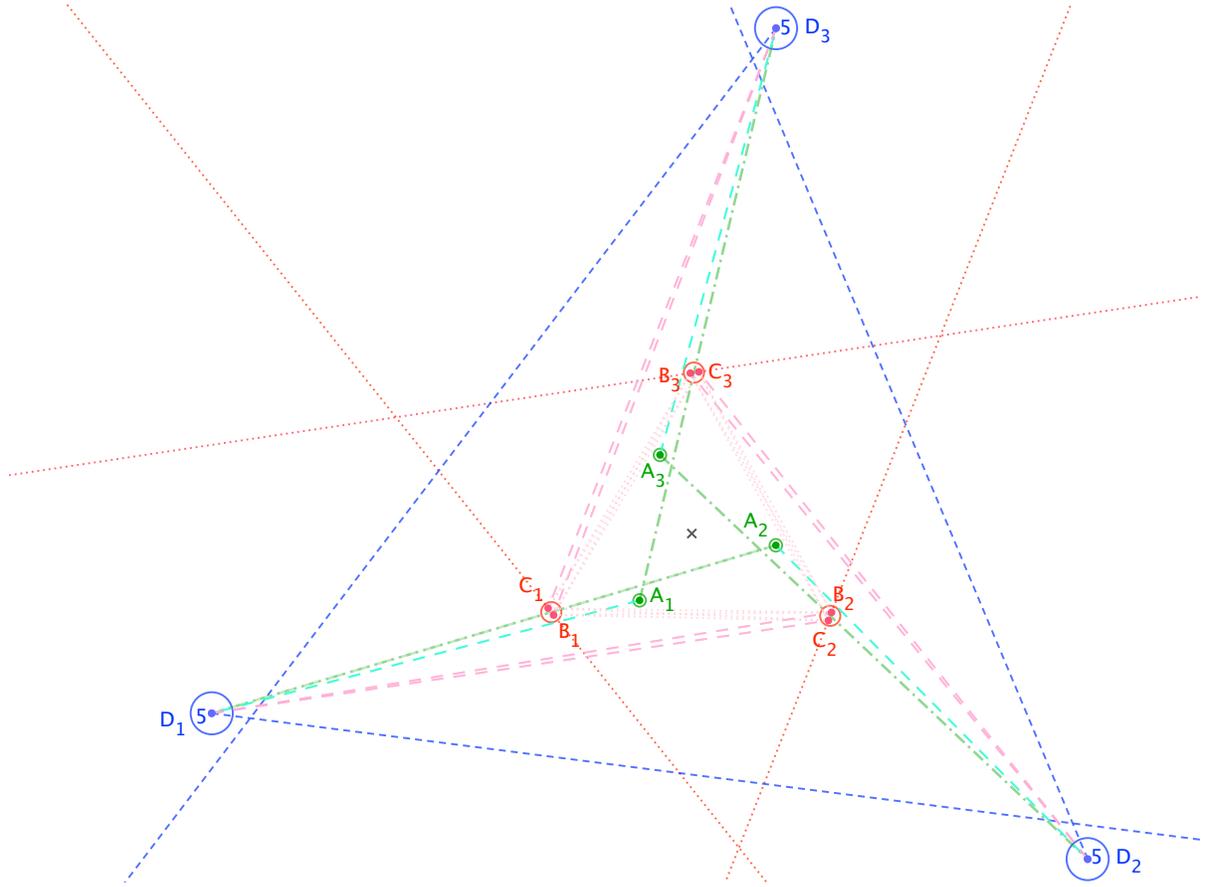}
\end{center}
\caption{A $3$-fold symmetric set of $24$ points with maximum crossing family of size $5$.}
\label{fig:5nover24}
\end{figure}

We show that the maximum crossing family of the point configuration described above is at most $5k$.

\begin{theorem}
There exist sets of $n$ points whose crossing family is of size at most $5\lceil\frac{n}{24}\rceil$.
\label{thm:up-5nover24}
\end{theorem}

Before proving Theorem~\ref{thm:up-5nover24}, we draw attention to the following lemma.

\begin{lemma}
If there exists an $m$-point set with maximum crossing family of size $f > 1$, then there exists an $n$-point set with maximum crossing family of size $ f\cdot\lceil \frac{n}{m} \rceil $.
\label{lem:grow}
\end{lemma}

\begin{proof}
Let $P$ be an $m$-point set with maximum crossing family of size $f$. We use $P$ as a base set to construct an $n$-point set $P'$ as follows: each point in $P$ is replaced with $\lfloor \frac{n}{m} \rfloor$ or $\lceil \frac{n}{m} \rceil$ imperceptibly perturbed copies (each copy of a point is distinct).
Let $\mathcal{F}'$ be a maximum crossing family in $P'$. Given $\mathcal{F}'$, contract each copy (of a point in $P$)  to the original point (in $P$) and let $Q$ be the new set of segments obtained. Note that the segments in $Q$ are all pairwise touching (that is, they either intersect or have a common endpoint).
Let $\mathcal{F} \subseteq Q$ be a maximum crossing family. For a set of segments $S$, let $p(S)$ denote the set of points induced by $S$.
If $|\mathcal{F}|>1$, then $Q$ admits no triangle that is formed by a point in $p(Q) \setminus p(\mathcal{F})$ and a segment in $\mathcal{F}$
(otherwise, at least one pair of segments in $Q$ are disjoint). Therefore, in this case, the maximum crossing family in $P'$ is at most $f\cdot\lceil\frac{n}{m}\rceil$. If $|\mathcal{F}|=1$, the size of the crossing family in $P'$ may be up to $\frac{3}{2}\lceil\frac{n}{m}\rceil$.

\end{proof}

\begin{proof}[Proof of Theorem~\ref{thm:up-5nover24}]
Let $P$ be the point configuration described earlier when $k=1$ (depicted in Figure~\ref{fig:5nover24}).
We show that the maximum crossing family of $P$ has size at most five. This, together with Lemma~\ref{lem:grow}, concludes the proof.

Recall that in the point configuration we described, the orientation of every triple of points from three different disks is the same as the orientation of their disks. This implies that our configuration (with $k>1$) is a subclass of the point sets obtained from appealing to Lemma~\ref{lem:grow}.


A \emph{disk-contraction} contracts all points in a disk to the center of the disk.
For a set $\mathcal{S}$ of segments, $\widetilde{\mathcal{S}}_{\mathcal{X}}$ denotes the set of segments obtained after performing a disk-contraction on $\mathcal{S}$, where each disk in $\mathcal{X}$ is contracted.
We refer to a segment $\overline{c(X)c(Y)}$ in $\widetilde{\mathcal{S}}_{\mathcal{X}}$, where $X,Y \in \mathcal{X}$, as an $XY$-segment or a segment of type $XY$. An $XY$-segment is incident to disks $X$ and $Y$.

Let $\mathcal{D}=\{A_i,B_i,C_i,D_i \mid i\in[3]\}$. Let $\mathcal{F}$ be a maximum crossing family of $P$. We first show that if $|\mathcal{F}| >5$, then $\mathcal{F}$ may only consist of segments of type $D_iX_j$, where $X \in \{A,B,C\}$. We then prove that any crossing family formed by these type of segments is indeed of size at most five.
We consider the following cases.
\begin{enumerate}[label={\bfseries Case \arabic*.},ref={\arabic*},leftmargin=2.8\parindent]
\item \label{c1}$\mathcal{F}$ is not incident to any $D$-disks. The number of points in $P$ that are not in $D$-disks is nine. Hence, in this case, the crossing family is of size at most four.
\item \label{c2}$\mathcal{F}$ contains a $D_iD_j$-segment (where $i,j \in [3]$). Recall that  the three sets $D_1,D_2$, and $D_3$ satisfy the full separating property and that $P \setminus \bigcup\limits_{i=1}^{3} D_i$ lies in the core defined by the $D$-disks. As a result, by Lemma~\ref{lem:nover4}, all segments in $\mathcal{F}$ are incident to $D_i$ or all are incident to $D_j$. Thus, the size of the crossing family is at most five.
\item \label{c3}$\mathcal{F}$ is incident to no $A$-disks. Let $\widetilde{\mathcal{F}}=\widetilde{\mathcal{F}}_{\{S_i,D_i \mid i \in[3]\}}$. Recall that the maximum crossing family of $\{c(D_i),c(S_i) \mid i \in[3] \}$ is of size two. Suppose $|\mathcal{F}|>5$. There are three possibilities:
\begin{itemize}
\item $\widetilde{\mathcal{F}}$ admits a triangle. By Case \ref{c2}, $\mathcal{F}$ contains no $D_iD_j$-segments. Hence, the triangle is incident to at least two $S$-disks. Note that $c(S_i)$ can be incident to at most two segments in $\widetilde{\mathcal{F}}$. If $\widetilde{\mathcal{F}}$ contains more than three segments, then a vertex of the triangle is incident to at least three segments (because the segments in $\widetilde{\mathcal{F}}$ are pairwise touching), and hence the triangle should be incident to one $D$-disc. Since $\widetilde{\mathcal{F}}$ is not incident to any $A$-disks, $\widetilde{\mathcal{F}}$ cannot contain a segment touching all three segments of the triangle. Hence, $\widetilde{\mathcal{F}}$ consists of exactly three segments, and $|\mathcal{F}| \le \frac{3}{2}\cdot 2=3$.
\item $\widetilde{\mathcal{F}}$ contains a pair of segments that are crossing. Note that $\widetilde{\mathcal{F}}$ admits no triangles and that $\mathcal{F}$ contains no $D_iD_j$-segments. As a result $|\mathcal{F}|\le 2 \cdot 2=4$.
\item $\widetilde{\mathcal{F}}$ forms a star (i.e. has a point that is incident to all segments). Thus, $|\mathcal{F}| \le 5$.
\end{itemize}
\item \label{c4}$\mathcal{F}$ contains an $A_iA_j$-segment (where $i,j \in [3]$).
Let $\widetilde{\mathcal{F}}=\widetilde{\mathcal{F}}_{\mathcal{D}}$. Let $e \in \widetilde{\mathcal{F}}$ be an $A_iA_j$-segment. All segments in $\widetilde{\mathcal{F}}$ are pairwise touching. The endpoints of $e$ cannot be incident to any other segments. Thus, each segment in $\widetilde{\mathcal{F}} \setminus e$ should intersect $e$. Moreover, any segment that crosses $e$ is incident to $c(A_k)$, where $k\neq i,j$. Thus, $|\mathcal{F}|=2$.
\item \label{c5}$\mathcal{F}$ contains an $S_iS_j$-segment (where $i,j \in [3]$).
Note that any segment crossing an $S_iS_i$-segment is of type $D_iA_{i+1}$. Thus, if $\mathcal{F}$ contains an $S_iS_i$-segment, $|\mathcal{F}|=2$.
Let $\widetilde{\mathcal{F}}=\widetilde{\mathcal{F}}_{\{A_i,S_i,D_i \mid i \in[3]\}}$.  Let $e \in \widetilde{\mathcal{F}}$ be an $S_iS_j$-segment, where $i \neq j$. Either $\widetilde{\mathcal{F}}$ admits a triangle or not. We consider the two cases below.
\begin{itemize}
\item $\widetilde{\mathcal{F}}$ admits a triangle. Suppose $|\mathcal{F}|>5$. By Case \ref{c2}, $\widetilde{\mathcal{F}}$ contains no $D_xD_y$-segments (where $x,y \in [3]$). For any $x \in[3]$, $c(A_x)$ can be incident to at most one segment in $\widetilde{\mathcal{F}}$. Thus, any triangle contains an $S_iS_j$-segment. Assume that the triangle is formed by $e$ and a point $v$. Recall that for $x \in[3]$, $S_x$ separates $D_x$ from $P \setminus (D_x \cup S_x)$. Thus $v \notin \{D_i,D_j\}$. Note that no segment incident to $v$ crosses $e$. Moreover, $c(S_i)$ and $c(S_j)$ can be incident to at most two segments. Therefore, $\widetilde{\mathcal{F}}$ consists of exactly three segments, and hence $|\mathcal{F}|\le \frac{3}{2}\cdot 2=3$.
\item $\widetilde{\mathcal{F}}$ does not admit a triangle. 
Any segment in $\widetilde{\mathcal{F}}$ that crosses $e$ is incident to an $A$-disk. Hence, $e$ can be crossed by at most three segments. At least one of $c(S_i)$ or $c(S_j)$ is incident to only one segment (if both are incident to two segments, then a triangle is formed). Therefore, $|\mathcal{F}|\le 5$.
\end{itemize}
\item \label{c6}$\mathcal{F}$ contains an $A_iS_j$-segment (where $i,j \in [3]$).
Let $e \in \mathcal{F}$ be an $A_iS_j$-segment. If $i=j$, then $|\mathcal{F}|=2$ because at most one segment can cross $e$. If $i \neq j$, then $|\mathcal{F}|\le 4$ because at most three segments can cross $e$.
\end{enumerate}
Suppose $|\mathcal{F}|>5$.
By Cases \ref{c1}-\ref{c6} , $\mathcal{F}$ may only contain segments of type $D_iX_j$, where $X \in \{A,B,C\}$ and $i,j \in [3]$. Thus, by Case \ref{c3}, there exists a segment of type $D_iA_j$. Let $\widetilde{\mathcal{F}}=\widetilde{\mathcal{F}}_{\mathcal{D}}$. We consider two cases.
\begin{itemize}
\item There exists a segment $e \in \widetilde{\mathcal{F}}$ that is of type $D_iA_i$ or $D_iA_{i+2}$.
Without loss of generality, assume $e$ is of type $D_iA_i$. (The other case is symmetric.) Let $Q \subseteq \widetilde{\mathcal{F}}$ be a maximum crossing family. Recall that all segments in $\widetilde{\mathcal{F}}$ are incident to $D$-disks. Note that all segments in $\widetilde{\mathcal{F}}$ that cross $e$ (if any) should be incident to the same $D$-disk. Thus, $|Q| \le 2$. All segments in $\widetilde{\mathcal{F}}$ are pairwise touching. Therefore, $\widetilde{\mathcal{F}}$ consists of either a star or two crossing stars (two stars are crossing if each segment of a star crosses all segments of the other star).
If $\widetilde{\mathcal{F}}$ forms a star, then $|\mathcal{F}| \le 5$. So, we assume that there exists a segment in $\widetilde{\mathcal{F}}$ that crosses $e$. Note that at most two such segments exist. Let $\mathcal{S}_1 \subset \widetilde{\mathcal{F}}$ denote the segments of the star containing $e$. Let $\mathcal{S}_2$ denote the segments of the star crossing $e$. $|\mathcal{S}_2| \le 2$. We consider the two cases below.
\begin{itemize}
\item $|\mathcal{S}_2| = 1$. Thus, there exists a $D_{i+1}S_i$-segment. Recall that $S_i$ separates $D_i$ from $D_{i+1}$. Hence, $\mathcal{S}_1$ is not incident to $S_i$. By Case \ref{c2}, there are no $D_iD_{i+1}$-segments. Hence $|\mathcal{S}_1| \le 4$, which implies $|\mathcal{F}| \le 5$.
\item $|\mathcal{S}_2| = 2$. Thus $|\mathcal{S}_1| \le 3$ (using Case \ref{c2}). Therefore, $|\mathcal{F}| \le 5$.
\end{itemize}
\item All segments in $\widetilde{\mathcal{F}}$ that connect a $D$-disk to an $A$-disk are of type $D_iA_{i+1}$. $\widetilde{\mathcal{F}}$ may be incident to one, two or three $A$-disks. We consider each of these cases below.
\begin{itemize}
\item $\widetilde{\mathcal{F}}$ is incident to exactly one $A$-disk.
Let $e$ denote the (only) $D_iA_{i+1}$-segment in $\widetilde{\mathcal{F}}$. If no segment in $\widetilde{\mathcal{F}}$ crosses $e$, then $|\mathcal{F}|\le 5$. So, we assume that there exists a segment in $\widetilde{\mathcal{F}}$ that crosses $e$. Note that at most one such segment exists. 
This implies that $\widetilde{\mathcal{F}}$ consists of a star together with a segment crossing all the star segments. Therefore, by Case \ref{c2} and the fact that $S_i$ separates $D_i$ from $D_{i+1}$, $|\mathcal{F}| \le 4$ (recall that $\widetilde{\mathcal{F}}$ is incident to only one $A$-disk).
\item $\widetilde{\mathcal{F}}$ is incident to exactly two $A$-disk. This implies that $\widetilde{\mathcal{F}}$ is incident to exactly two $D$-disks (otherwise, either $\widetilde{\mathcal{F}}$ is not pairwise touching or it is incident to all $A$-disks). Thus, $\widetilde{\mathcal{F}}$ consists of two crossing stars. Note that the number of segments in each star is at most two. Thus, $|\mathcal{F}|\le 4$.
\item $\widetilde{\mathcal{F}}$ is incident to three $A$-disk. Let $E$ denote the set of three segments in $\widetilde{\mathcal{F}}$ incident to $A$-disks. If $|\mathcal{F}|>3$, there exists a segment in $\widetilde{\mathcal{F}}$ that crosses at least two of the segments in $E$. Any such segment is incident to an $A$-disk. Recall that no $A$-disk can be incident to more than one segment. Therefore, $|\mathcal{F}|=3$.
\end{itemize}
\end{itemize}
Therefore, the maximum crossing family of $P$ cannot be greater than five.
\end{proof}

\section{$1$-avoiding point sets}\label{sec:sc}

In this section, we restrict our attention to $1$-avoiding point sets.
We study the combinatorial properties of crossing families in $1$-avoiding point sets. We then introduce a relaxation on these combinatorial properties.
As a step towards finding maximum crossing families in $1$-avoiding point sets, we study the size of our relaxed notion of crossing families called ``side compatible subsets''. We give linear bounds on side compatible subsets for some simplified versions of the problem.

\subsection{Combinatorial properties of crossing families}

Let $P=P_{B \vdash R}$ be a $1$-avoiding point set $B \cup R$ where $B$ and $R$ are two separable equal-sized sets of points such that $B$ avoids $R$.
Without loss of generality, for point set $P_{B \vdash R}$, we assume $B$ and $R$ are separable by a vertical line and $R$ lies to the left of $B$.
The \emph{dual} of a point $p=(a,b)$ is the line $p^\star=\{(x,y) \mid y=ax-b\}$ and the dual of a line $l=\{(x,y) : y=ax+b\}$ is the point $l^\star=(a,-b)$.
The dual of a $1$-avoiding point set $P_{B \vdash R}$ consists of two sets of lines $B^\star$ (blue lines) and $R^\star$ (red lines) such that each line in $B^\star$ has a larger slope than any line in $R^\star$ and each line in $R^\star$ intersects (or sees) the lines of $B^\star$ in the same order.
We call such a line arrangement a \emph{$1$-avoiding} line arrangement.

We study characteristics of crossing families in $1$-avoiding point sets in combinatorial terms and in the dual plane.

\begin{definition}
An \emph{allowable sequence} is a sequence $\pi_1, \dots, \pi_l$ of permutations of $[n]$ satisfying the following properties:
\begin{enumerate}
\item $\pi_1$ is the identity permutation (i.e. $\pi_1=1,\dots,n$) and $\pi_l$ is the reverse of $\pi_1$;
\item The move from $\pi_i$ to $\pi_{i+1}$ consists of reversing one or more non-overlapping
substrings (each of size at least two);
\item Any two elements in $[n]$ reverse their order exactly once.
\end{enumerate}
An allowable sequence is \emph{simple} if each move from $\pi_i$ to $\pi_{i+1}$ consists of reversing just one pair of elements.
\end{definition}

All combinatorial information of a point configuration or a line arrangement in general position is encoded by a simple allowable sequence.
A simple allowable sequence associated with a point configuration encodes the total slope order of lines joining pairs of points. An allowable sequence associated with a line arrangement $L$ encodes the order in which a sweeping vertical line from left to right sees the intersection points of arrangement $L$.
Using allowable sequences for studying a geometric problem on a point configuration or a line arrangement is equivalent to generalizing the problem to a generalized point configuration\footnote{A \emph{generalized configuration} of points in general position is a finite set of points in the Euclidean plane, together with an arrangement of pseudolines, such that every pair of points lie on exactly one pseudoline and each pseudoline contains exactly two points.} or a pseudoline arrangement.


Let $|B^\star|=|R^\star|=n$. Say each line in $R^\star$ sees the blue lines in the order $b^\star_1,b^\star_2,\dots,b^\star_n$; and $b^\star_1$ intersects the red lines in the order $r^\star_1,r^\star_2,\dots,r^\star_n$. We can represent the order in which each blue line intersects the red lines in a table, where row $i$ contains the indices of red lines in the order that they are seen by $b^\star_i$.
We denote such a table by $T({B^\star},{R^\star})$. For subsets ${B^\star}' \subset B^\star$ and ${R^\star}' \subset R^\star$, we can \emph{reduce} the table to ${B^\star}'$ and ${R^\star}'$ if we only keep the rows representing ${B^\star}'$ and indices representing ${R^\star}'$. We refer to the table reduced to ${B^\star}'$ and ${R^\star}'$ as a \emph{subtable}.

We know $B$ and $R$ can be crossed if and only if  they obey the rank condition. This implies that $P_{B \vdash R}$ has a crossing family of size $k$ whose segments connect a point in $B$ to a point in $R$, if and only if there are subsets ${B^\star}' \subseteq B^\star$ and ${R^\star}'\subseteq R^\star$, each of size $k$, such that the elements of the diagonal of the table reduced to ${B^\star}'$ and ${R^\star}'$ are all distinct.

\begin{lemma}
Let $P=P_{B \vdash R}$. A subtable of $T(B^\star,R^\star)$ whose diagonal $\mathrm{d}_1 \mathrm{d}_2 \dots \mathrm{d}_k$ consists of distinct elements has the property that
for all rows $j$, $\mathrm{d}_i$ is before $\mathrm{d}_j$ if and only if $i<j$.
See Figure~\ref{fig:tab}.
\label{lem:tab}
\end{lemma}
\begin{figure}[h]
\begin{center}
\begin{tabular}{c|c|c|c|c|c|c|c|c|}
\cline{2-9}
\small{${b^\star}'_1$} & \multicolumn{1}{l}{$\mathrm{\mathbf{d_1}}$ } & \multicolumn{1}{c}{}  &  \multicolumn{5}{r}{$\{\mathrm{d_2} ~ \dots ~ \mathrm{d_k}\}$}&\\
\small{${b^\star}'_2$} & \multicolumn{1}{l}{$\mathrm{d_1}$} & \multicolumn{1}{c}{$\mathrm{\bf{d_2}}$}  &  \multicolumn{5}{r}{$\{\mathrm{d_3} ~ \dots ~ \mathrm{d_k}\}$} & \\
\small{${b^\star}'_3$} & \multicolumn{2}{l}{$\{\mathrm{d_1} ~  \mathrm{d_2}\}$} & \multicolumn{1}{c}{$\mathrm{\bf{d_3}}$} & \multicolumn{4}{r}{$\{\mathrm{d_4} ~ \dots ~ \mathrm{d_k}\}$} & \\
 &\multicolumn{7}{c}{\multirow{4}{*}{}}& \\
 & \multicolumn{7}{c}{}&\\
 & \multicolumn{7}{c}{}&\\ 
  & \multicolumn{7}{c}{}&\\
\small{${b^\star}'_{k-1}$} & \multicolumn{5}{l}{$\{\mathrm{d_1} ~ \dots ~ \mathrm{d_{k-2}}\}$} & \multicolumn{1}{r}{$\hspace{1.2cm}\mathrm{\bf{d_{k-1}}}$}  &  \multicolumn{1}{r}{$\mathrm{d_k}$} & \\
\small{${b^\star}'_k$} & \multicolumn{6}{l}{$\{\mathrm{d_1} ~ \mathrm{d_2} ~ \dots ~ \mathrm{d_{k-1}}\}$} & \multicolumn{1}{r}{$\mathrm{\bf{d_k}}$}  &\\
\cline{2-9}
\end{tabular}
\end{center}
\caption{The structure of a subtable with distinct diagonal entries.}
\label{fig:tab}
\end{figure}

\begin{proof}
We can think of $P^\star$ as a set of lines $R^\star$ which are seen by a set of parallel lines $B^\star$. Hence the permutations obtained from the rows of $T(B^\star,R^\star)$ satisfy the property that each permutation is either identical to the preceding permutation or can be obtained from it by reversing one or more non-overlapping increasing substrings. Note that no two elements reverse their order more than once in these permutations (i.e. no two lines cross more than once).

We prove by contradiction. Let row $b$ be the first row where there exists $a<b$ such that $\mathrm{d}_a$ is after $\mathrm{d}_b$. Note that $b>1$. Let row $c$ be the last row where $\mathrm{d}_a$ is after $\mathrm{d}_c$. Note that $b \le c <k$. Since there are $c-1$ elements before $\mathrm{d}_c$ in row $c$ and $\mathrm{d}_a$ is after $\mathrm{d}_c$, there must be an element $\mathrm{d}_d$, with $d>c$, before $\mathrm{d}_c$ in row $c$. Thus, 
\begin{itemize}
\item in row $a$, $\mathrm{d}_a$ is before $\mathrm{d}_d$;
\item in row $c$, $\mathrm{d}_a$ is after $\mathrm{d}_d$; and
\item in row $d$, $\mathrm{d}_a$ is before $\mathrm{d}_d$;
\end{itemize}
which is not possible since $\mathrm{d}_a$ and $\mathrm{d}_d$ cannot reverse their order more than once.
\end{proof}

Lemma~\ref{lem:tab} implies the following.

\begin{corollary}
Let $P=P_{B \vdash R}$. Let $T'$ be a subtable of $T(B^\star,R^\star)$ with distinct diagonal entries. Let $T'$ be obtained by reducing $T(B^\star,R^\star)$ to ${B^\star}'$ and ${R^\star}'$, where $|{B^\star}'|=|{R^\star}'|$. Let ${B^\star}'={{b^\star}'_1,\dots {b^\star}'_k}$. Let $\mathrm{d}_i$ and $\mathrm{d}_j$ (where $i \neq j$) be two elements on the diagonal of $T'$. The intersection of red lines $r^\star_{\mathrm{d}_i}$ and $r^\star_{\mathrm{d}_j}$ is either above both blue lines ${b^\star}'_i$ and ${b^\star}'_j$, or below both of them.
\label{cor:crf-tab}
\end{corollary}

Figure~\ref{fig:tab-line} illustrates an example of a $1$-avoiding line arrangement together with its table representation.

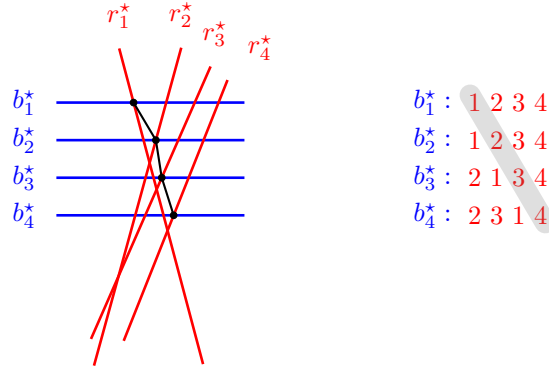
\begin{figure}[h]
\begin{center}
\begin{tikzpicture}[scale=.5]
	\def \lw {1};
	\def \d {5};

	\foreach \i in {1,2,3,4}
		\draw[blue,line width=\lw pt,name path=B\i] (0,{6-\i}) -- (-\d,{6-\i}) node [label=left:{$b^\star_\i$}] {};
	
	\path[draw,red,line width=\lw pt, name path=R1] (-1.1,-1.96) -- (-3.33, 6.44) node [label=above:{$r^\star_1$}] {};
	\path[draw,red,line width=\lw pt, name path=R2] (-4,-2) -- (-1.68,6.46) node [label=above:{$r^\star_2$}] {};
	\path[draw,red,line width=\lw pt, name path=R3] (-4.08,-1.29) -- (-0.9,5.96) node [label=above:{$~r^\star_3$}] {};
	\path[draw,red,line width=\lw pt, name path=R4] (-3.21,-1.34) -- (-0.45,5.6) node [label=above right:{$r^\star_4$}] {};

	\path [name intersections={of=R1 and B1,by=D1}];
	\path [name intersections={of=R2 and B2,by=D2}];
	\path [name intersections={of=R3 and B3,by=D3}];
	\path [name intersections={of=R4 and B4,by=D4}];			

	\foreach \i in {1,2,3,4}{
		\draw[fill] (D\i) circle (2.5*\lw pt);
	}
	
	\draw[ thick] (D1) -- (D2) -- (D3) --(D4);

	\foreach \i in {1,2,3,4}{
		\node[blue] at (\d,{6-\i}) {{$b^\star_\i:$}};	
	}

	\node[red] at ({1.4*\d},{5}) {{$1 ~ 2 ~ 3 ~ 4$}};	
	\node[red] at ({1.4*\d},{4}) {{$1 ~ 2 ~ 3 ~ 4$}};	
	\node[red] at ({1.4*\d},{3}) {{$2 ~ 1 ~ 3 ~ 4$}};	
	\node[red] at ({1.4*\d},{2}) {{$2 ~ 3 ~ 1 ~ 4$}};				
	
	\draw[line cap=round,gray,nearly transparent, line width=8pt] (1.2*\d,5.2)--(1.6*\d,1.8);
		
\end{tikzpicture}

\end{center}
\caption{The polygonal line in the arrangement corresponds to the diagonal in the table.}
\label{fig:tab-line}
\end{figure}

Using the table representation $T(B^\star,R^\star)$, it is easy to see that any $1$-avoiding $2n$-point set $P=P_{B \vdash R}$ has a crossing family of size at least $\sqrt{n}$. The middle row of the table has $n$ elements, which implies there is a subsequence of size at least $\sqrt{n}$ in the middle row which is either increasing or decreasing.  Denote this monotone subsequence by $S=\{s_1,s_2,\dots, s_l\}$. Let $R^\star_S=\{r^\star_{s_i} \mid s_i \in S\}$. If the subsequence is increasing, then all the red lines in $R^\star_S$ intersect after the middle blue line, and hence $T(B^\star,R^\star)$ can be reduced to $R^\star_S$ and any $l$-subset of $\{b^\star_1,b^\star_2,\dots,b^\star_{\frac{n}{2}}\}$ to form a subtable with distinct diagonal entries. If the subsequence is decreasing, all red lines in $R^\star_S$ intersect before the middle blue line, and hence there exists a subtable with distinct diagonal entries among the lower half of $T(B^\star,R^\star)$.

In the following, we study a generalized variant of a table, called a combinatorial table, where
each row is a permutation of $[n]$ and no two elements reverse their order more than once. The sequence of the permutations of a combinatorial table may contain identical permutations and any subsequence containing distinct permutations of a combinatorial table is a subsequence of an allowable sequence. A combinatorial table corresponds to a table representation of a $1$-avoiding arrangement which may not be realizable using straight lines (i.e. $1$-avoiding arrangement of pseudolines).

\begin{definition}
Let $\Pi=\Pi_1\Pi_2\cdots\Pi_l$ be a sequence of $l$ permutations of $[n]$ such that no two elements reverse their order more than once. Let $S=\{s_1,\ldots,s_m\}$ be an $m$-subset of $[l]$ where $s_1<\dots<s_m$. Let $E$ be an $m$-subset of $[n]$. A \emph{combinatorial} table $\mathtt{T}(S,E)$ is a table where row $i$ is the subsequence of $\Pi_{s_i}$ containing only the elements that are in $E$. We may refer to a combinatorial table as a table for simplicity.
\end{definition}

\begin{lemma}
Let $\Pi=\Pi_1\Pi_2\cdots\Pi_{{n \choose 2} +1}$ be a simple allowable sequence with $\Pi_1=1,2,\dots,n$. There exits $S=\{s_1,\dots, s_{n-1}\}$ where  $S \subset [{{n \choose 2}+1}]$ and $s_1<\dots<s_{n-1}$, such that table $\mathtt{T}(S,[2..n])$ has distinct entries on its diagonal.
\end{lemma}

\begin{proof}
Let $S$ be the set of all numbers $i < {n \choose 2}+1$ where $\Pi_i$ differs from $\Pi_{i+1}$ in having $1$ and its right neighbour flipped. Let $\mathrm{d}_1,\dots,\mathrm{d}_{n-1}$ be the diagonal entries of $\mathtt{T}(S,[2..n])$. Note that $1$ is the leftmost element in $\Pi_1$ and no pair of elements flip more than once. Thus, $\mathrm{d}_i$ is the $i$-th element flipped by $1$. If an entry $\mathrm{d}_j$ on the diagonal is repeated, then $1$ and $\mathrm{d}_j$ have to flip more than once in the allowable sequence, which is not possible.
\end{proof}

\begin{lemma}
Let $\Pi=\Pi_1\Pi_2\cdots \Pi_l$ be a sequence of permutations of $[n]$ where $\Pi$ may contain repeated permutations but all repetitions of the same permutation are consecutive, and the largest subsequence containing all distinct permutations is a subsequence of an allowable sequence. There exists $\Pi$ with $l=\Theta(n^2)$, such that for no $n$-subset $S$ from $[l]$, table $T(S,[n])$ has distinct elements on its diagonal.
\end{lemma}

\begin{proof}[Proof sketch]
Let $l=\lfloor \frac{n}{2} \rfloor^2+n-1$. Construct $\Pi=\Pi_1 \Pi_2 \cdots \Pi_l$ such that $\Pi_1= 1 2 \cdots n$, and for all $1 \le i < l$, $\Pi_{i+1}$ is the same as $\Pi_i$ except for the following cases:
\begin{itemize}
\item if $i=\lceil \frac{n}{2} \rceil +k(n+1)$, where $k \in \{0,1,\dots,\lceil \frac{n-1}{4} \rceil-1\}$, the permutation $\Pi_{i+1}$ differs from $\Pi_{i}$ in having $\frac{n}{2}-k$ and $\frac{n}{2}+1+k$ flipped, and 
\item if $i=n+k(n+1)$, where $k \in \{0,1,\dots,\lfloor \frac{n}{4} \rfloor-1\}$, the permutation $\Pi_{i+1}$ differs from $\Pi_{i}$ in having $1+k$ and $n-k$ flipped.
\end{itemize}
\end{proof}

\subsection{Relaxation on the combinatorial properties}
As a step towards finding maximum crossing families in $1$-avoiding point sets, we study a relaxed notion of crossing families.

Let $\text{crf}(P_1,P_2)$ denote the size of the maximum crossing family whose segments connect a point in $P_1$ to a point in $P_2$.
Let $B=\{\mathtt{b}_1,\dots,\mathtt{b}_n\}$ and $R=\{\mathtt{r}_1,\dots,\mathtt{r}_n\}$ be two separable sets of $n$ points such that $B$ avoids $R$. Let $B^\star=\{\mathtt{b}^\star_1,\dots,\mathtt{b}^\star_n\}$ and $R^\star=\{\mathtt{r}^\star_1,\dots,\mathtt{r}^\star_n\}$ denote the duals of $B$ and $R$, respectively.
Recall that $\text{crf}(B,R)=k$ if and only if there are $k$-subsets ${B^\star}' \subset B^\star$ and ${R^\star}' \subset R^\star$ such that $T'=T({B^\star}',{R^\star}')$ has distinct diagonal entries. A necessary property for $T'$ (by Corollary~\ref{cor:crf-tab}) is that for any two diagonal entries $\mathrm{d}_i$ and $\mathrm{d}_j$ of $T'$, the intersection of red lines $r^\star_{\mathrm{d}_i}$ and $r^\star_{\mathrm{d}_j}$ is either above both blue lines ${b^\star}'_i$ and ${b^\star}'_j$, or below both of them. We refer to this property as the ``sidedness'' property in crossing families.
In the following, we introduce a relaxation on crossing families, called ``side compatibility'', which preserves the sidedness property.

We start by defining some terminology. (See Figure~\ref{fig:s-vs-os} for illustrations.)
Given a line arrangement $\mathcal{L}=\{l_1,\dots,l_n\}$, the \emph{bar representation} of $\mathcal{L}$ is composed of $n \choose 2$ horizontal \emph{bars} (i.e. horizontal segments) arranged such that the $i$-th bar from below represents the $i$-th intersection point $p_i$ from below in $\mathcal{L}$. We assume no intersection points in $\mathcal{L}$ have the same $y$-coordinate (we rotate $\mathcal{L}$ if necessary).
For intersection point $p_i$ of lines $l_{a_i}$ and $l_{b_i}$ in $\mathcal{L}$, the corresponding bar is a segment from $(a_i,i)$ to $(b_i,i)$. 

A \emph{bar stack} $\mathcal{B}_{l,n}$ is an arrangement of $l$ bars $B_1 B_2 \cdots B_l$ where ($\expandafter{\romannumeral 1}$) each bar $B_i$ extends from $(a_i,i)$ to $(b_i,i)$ where $a_i$ and $b_i$ are integers and $1 \le a_i < b_i \le n$,
and ($\expandafter{\romannumeral 2}$) for two bars $B_i$ and $B_j$ either $a_i \neq a_j$ or $b_i \neq b_j$ or both.
We say $B_i$ is at \emph{height} $i$. Note that $\mathcal{B}_{l,n}$ may not come from the representation of a line arrangement.
This makes bar stacks more expressive than line arrangements at representing ``intersections''. However, bar stacks may represent fewer ``intersections'' compared to line arrangements. Having fewer than $n \choose 2$ bars (``intersections'') makes bar stacks easier to work with.

Let $\mathcal{W}_{n,l}=\mathit{w}_1 \mathit{w}_2 \dots \mathit{w}_n$ be a sequence of $n$ integers from zero to $l$ (the numbers may not be distinct).
Let $\mathcal{C}$ be a subset of $[n]$ (containing distinct numbers) for which there exists an injective function $f: \mathcal{C} \rightarrow \mathcal{W}_{n,l}$ such that for any bar $B_i$ whose endpoints' $x$-coordinates both belong to $\mathcal{C}$, either
\begin{enumerate}[label={(\bfseries P\arabic*)}]
\item $f(a_i) < i$ and $f(b_i) <i$\label{bar-P1}, or
\item $f(a_i) \ge  i$ and $f(b_i) \ge i$\label{bar-P2}. 
\end{enumerate}
We say $\mathcal{C}$ is a \emph{side compatible subset} for $\mathcal{W}_{n,l}$ in $\mathcal{B}_{l,n}$. We refer to the function $f$ as the \emph{mapping function} for $\mathcal{C}$.
If for a side compatible subset $\mathcal{C}$, we have the additional property that in case of \ref{bar-P1}, $f(a_i)<f(b_i)$ and in case of \ref{bar-P2}, $f(a_i)>f(b_i)$, we say $\mathcal{C}$ is an \emph{ordered} side compatible subset.


The mapping function $f$ of a side compatible subset $\mathcal{C}$ matches every element $c \in \mathcal{C}$ to $f(c)$; thus, the pair $(\mathcal{C},f)$ defines a matching with edges $\{(c,f(c)) \mid c \in \mathcal{C}\}$. A side compatible subset satisfies the sidedness property, in the sense that for $c_i,c_j \in  \mathcal{C}$, where $c_i<c_j$,  the height of the bar with horizontal interval $[c_i,c_j]$ (representing the ``intersection'' of $c_i$ and $c_j$) is less than both $f(c_i)$ and $f(c_j)$, or greater than or equal to both of them.

We can visualize a side compatible subset in the following way. Think of each $\mathit{w}_i \in \mathcal{W}_{n,l}$ as a distinct horizontal wire above $B_{\mathit{w}_i}$ and below $B_{\mathit{w}_i+1}$. If $\mathit{w}_i =0$ the wire is below $B_1$, and if $\mathit{w}_i = l$ the wire is above $B_l$. ($\mathcal{W}_{n,l}$ is represented by $n$ different horizontal wires.) We refer to the vertical lines going through endpoints of bars as \emph{pillars}. The pillar through $e \in [n]$ is the vertical line $x=e$.
We may think of function $f:\mathcal{C} \rightarrow \mathcal{W}_{n,l}$ as assigning a marble for each $\mathit{c} \in \mathcal{C}$ on the intersection point of the pillar through $c$ and wire $f(c)$.
A \emph{marbling} of \emph{size} $m$ is a set of $m$ marbles such that all marbles lie on the intersection points of the pillars  and the wires. A \emph{valid} marbling is a marbling such that each pillar or wire contains at most one marble.
We say an endpoint of a bar is \emph{associated} with a marble if the vertical line (pillar) through it contains a marble. (A pillar may go through the endpoints of a number of bars.)
A \emph{side compatible marbling} is a valid marbling such that for every bar both of whose endpoints are associated with marbles, both marbles are above or below the bar.
An ordered side compatible subset corresponds to a side compatible marbling such that for every bar both of whose endpoints are associated with marbles, the marble associated with the right endpoint is closer to the bar. We refer to such a marbling as an \emph{ordered side compatible marbling}.

Let $P_{B \vdash R}$ be a $1$-avoiding $2n$-point set, where $B=\{\mathtt{b}_1,\dots,\mathtt{b}_n\}$ and $R=\{\mathtt{r}_1,\dots,\mathtt{r}_n\}$. $B^\star=\{\mathtt{b}^\star_1,\dots,\mathtt{b}^\star_n\}$ and $R^\star=\{\mathtt{r}^\star_1,\dots,\mathtt{r}^\star_n\}$ denote the duals of $B$ and $R$, respectively.
Let $\mathcal{W}=\mathit{w}_1,\dots,\mathit{w}_n$, where $\mathit{w}_i$ is the number of intersection points among $R^\star$ that are below line $\mathtt{b}^\star_i$.
Consider a side compatible marbling corresponding to a side compatible subset $\mathcal{C}$ for $\mathcal{W}$ in the bar representation of $R^\star$ using the mapping function $f$.
A marble at the intersection point of the pillar through $c$ and wire $f(c)$ corresponds to the segment $\mathtt{r}_c\mathtt{B}_{f(c)}$.
This correspondence implies that any crossing family defines a marbling that is side compatible (by Corollary~\ref{cor:crf-tab}). So clearly, $\text{crf}(B,R)=k$ implies that there is a side compatible subset of size $k$ for $\mathcal{W}$ in the bar representation of $R^\star$. However, 
the reverse is not true. See Figure~\ref{fig:s-vs-os} for an example.



\begin{figure}[h]
\begin{center}
\begin{tikzpicture}[scale=.5]
	\def \lw {1};
	\def \d {5};

\begin{scope}[shift={(-.5*\d,0)}]

	\foreach \i in {1,2,3}
		\draw[blue,line width=\lw pt,name path=B\i] (.5*\d,\i) -++ (-\d,0) node [label=left:{$b^\star_\i$}] {};
	
	\path[draw,red,line width=\lw pt, name path=R1]  (-0.04, -0.99) -- (1.22,4.92)node [label=above left:{$r^\star_1$}] {};
	\path[draw,red,line width=\lw pt, name path=R2]  (-1.73,-1) -- (1.57,4.97) node [label=above:{$~r^\star_2$}] {};
	\path[draw,red,line width=\lw pt, name path=R3]  (-2.65,-1.04) -- (2.12,4.9) node [label=above right:{$r^\star_3$}] {};
			
\end{scope}
			
	\foreach \i in {1,2,3}{
		\draw[blue,line width=\lw pt] (2*\d,{\i}) -- (\d,{\i}) ;
		\draw[red,dotted,line width=\lw pt] ({\d+\i},0) -- ({\d+\i},5) ;
	}

	\draw[fill] (\d+1,3) circle (3.5*\lw pt);	
	\draw[fill] (\d+2,1) circle (3.5*\lw pt);	
	\draw[fill] (\d+3,2) circle (3.5*\lw pt);	

	\draw[red, line width=\lw pt] (\d+2,2.5) -- (\d+3,2.5) ;
	\draw[red, line width=\lw pt] (\d+3,3.5) -- (\d+1,3.5) ;
	\draw[red, line width=\lw pt] (\d+2,4.5) -- (\d+1,4.5) ;

	\draw[red,dotted,line width=.5 pt] ({\d-1},2.5) -++ ({\d+2},0);
	\node[red] at ({2*\d+3},2.5) {$y=1$};
	\draw[red,dotted,line width=.5 pt] ({\d-1},3.5) -++ ({\d+2},0);
	\node[red] at ({2*\d+3},3.5) {$y=2$};
	\draw[red,dotted,line width=.5 pt] ({\d-1},4.5) -++ ({\d+2},0);
	\node[red] at ({2*\d+3},4.5) {$y=3$};

 \node	at (1.5*\d,-1) {$\mathcal{W}=\langle 0, 0, 1 \rangle$};
 
\end{tikzpicture}

\end{center}
\caption{The left figure shows a $1$-avoiding line arrangement $B^\star \cup R^\star$, where $\text{crf}(B,R)=2$. On the right, the bar representation of $R^\star$ together with the wires (blue) corresponding to $\mathcal{W}=\langle 0,0,1 \rangle$ (representing $B^\star$) is depicted. The dots represent a side compatible marbling of size three. Note that the largest \emph{ordered} side compatible marbling is of size two.}
\label{fig:s-vs-os}
\end{figure}
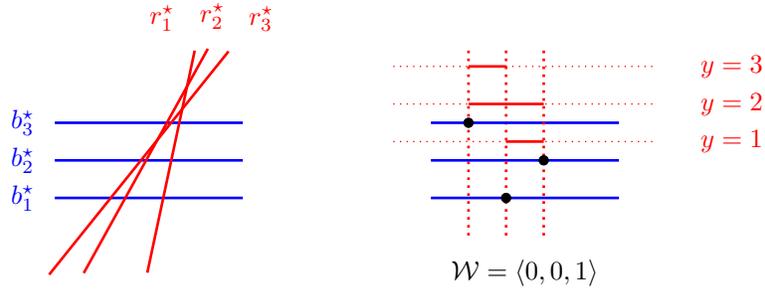

The sidedness property implies that the intersection point of the supporting lines of the two segments corresponding to two marbles in a side compatible marbling is either on both segments (i.e., the segments are crossing) or on neither of them. This is easy to verify under duality. See Figure~\ref{fig:s-os-dual}.

\begin{figure}[h]
\begin{center}
\begin{tikzpicture}[scale=.5]
	\def \lw {1};
	\def \d {8};

	\foreach \j in {1,2,3,4}{
		\foreach \i in {1,2}{
			\draw[blue,line width=\lw pt,name path=B\j\i] (\j*\d,\i) -++ (.5*\d,0);
		}

		\pgfmathparse{Mod(\j,2)==0?1:0}
		\ifnum \pgfmathresult=0
			\path[draw,red,line width=\lw pt, name path=R1]  (\j*\d,3) -- (\j*\d+2,0);
			\path[draw,red, line width=\lw pt, name path=R2]  (4+\j*\d,3) -- (\j*\d+1.5,0);
		\fi

		\ifnum \pgfmathresult>0
			\path[draw,red,line width=\lw pt, name path=R1]  (\j*\d+2,3) -- (\j*\d,0);
			\path[draw,red,line width=\lw pt, name path=R2]  (\j*\d+1.5,3) -- (\j*\d+4,0);
		\fi
		
		\path [name intersections={of=B\j2 and R1,by=M\j21}];
		\path [name intersections={of=B\j1 and R2,by=M\j12}];				
		\path [name intersections={of=B\j1 and R1,by=M\j11}];
		\path [name intersections={of=B\j2 and R2,by=M\j22}];				

 	}
 	
	\foreach \k in {121,112,221,212,311,322,411,422}
		\draw[fill] (M\k) circle (4pt);
	
	\fill[opacity=.2,green] (M121) -- (\d,3) -- (\d,2) -- (M121);
	\fill[opacity=.2,green] (M121) -- (1.5*\d,2) -- (1.5*\d,0) -- (\d+2,0) -- (M121);
	\fill[pattern=north west lines,opacity=.5] (M112) -- (1.5*\d,3) -- (\d,3) -- (\d,1) -- (M112);
	\fill[pattern=north west lines,opacity=.5] (M112) -- (1.5*\d,1) -- (1.5*\d,0)  -- (\d+1.5,0) -- (M112);
	
	\fill[opacity=.2,green] (M221) -- (2*\d+2,3) -- (2*\d,3) -- (2*\d,2) -- (M221);
	\fill[opacity=.2,green] (M221) -- (2.5*\d,2) -- (2.5*\d,0) -- (2*\d,0) -- (M221);	
	\fill[pattern=north west lines,opacity=.5] (M212) -- (2*\d+1.5,3) -- (2*\d,3) -- (2*\d,1) -- (M212);
	\fill[pattern=north west lines,opacity=.5] (M212) -- (2.5*\d,1) -- (2.5*\d,0)  -- (M212);		
			
	\fill[opacity=.2,green] (M311) -- (3*\d,3) -- (3*\d,1) -- (M311);
	\fill[opacity=.2,green] (M311) -- (3.5*\d,1) -- (3.5*\d,0) -- (3*\d+2,0) -- (M311);
	\fill[pattern=north west lines,opacity=.5] (M322) -- (3.5*\d,3) -- (3*\d,3) -- (3*\d,2) -- (M322);
	\fill[pattern=north west lines,opacity=.5] (M322) -- (3.5*\d,2) -- (3.5*\d,0)  -- (3*\d+1.5,0) -- (M322);
			
	\fill[opacity=.2,green] (M411) -- (4*\d+2,3) -- (4*\d,3) -- (4*\d,1) -- (M411);
	\fill[opacity=.2,green] (M411) -- (4.5*\d,1) -- (4.5*\d,0) -- (4*\d,0) -- (M411);	
	\fill[pattern=north west lines,opacity=.5] (M422) -- (4*\d+1.5,3) -- (4*\d,3) -- (4*\d,2) -- (M422);
	\fill[pattern=north west lines,opacity=.5] (M422) -- (4.5*\d,2) -- (4.5*\d,0)  -- (M422);		
						
\end{tikzpicture}
\end{center}
\caption{The shaded double-wedges show the corresponding segments of two marbles from a side compatible marbling in the dual plane. Only the first two pairs of marbles (on the left) are ordered side compatible.}
\label{fig:s-os-dual}
\end{figure}
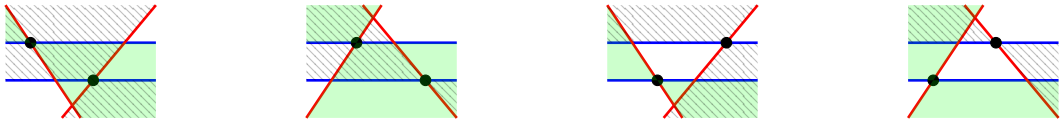

Note that the pair $(\mathcal{B}_{l,n}, \mathcal{W}_{n,l})$ defines a sequence of permutations of $[n]$.
In particular, we can assign a permutation to each wire $w$ as follows: we start with the identity permutation, and consider a horizontal sweep line that moves top-down until it hits wire $w$; whenever the sweep line hits a bar with horizontal interval $[i,j]$, we swap the positions of $i$ and $j$ in the permutation.
For a side compatible subset $\mathcal{C}$, let $\mathcal{M}_\mathcal{C}$ denote the corresponding side compatible marbling for $\mathcal{C}$.
Let $T(\mathcal{C})$ denote the table where row $i$ is the permutation, restricted to elements of $\mathcal{C}$, assigned to the $i$-th wire (numbered top-down) containing a marble.
In order for a side compatible subset $\mathcal{C}$ to correspond to a crossing family, the diagonal entries of $T(\mathcal{C})$ should all be distinct. It is easy to verify that the diagonal entries are distinct only if $\mathcal{C}$ is an ordered side compatible subset; or in other words,
for the segments corresponding to a side compatible marbling to cross pairwise, the marbling should be ordered side compatible. This is easy to see under duality because if two marbles satisfy side compatibility but not ordered side compatibility, then the supporting lines of the corresponding two segments intersect outside both of them.
Thus, $\text{crf}(B,R)=k$ if and only if the largest \emph{ordered} side compatible subset for $\mathcal{W}$ in the bar representation of $R^\star$ is of size $k$.
Hence, $\min_B \text{crf}(B,R) = k$ (where the minimum is taken over all point sets $B$ that avoid $R$) if and only if for any $\mathcal{W}_{n,{n \choose 2}}$, there exists an ordered side compatible subset $\mathcal{C}$ in the bar representation of $R^\star$, where $|\mathcal{C}|=k$.

\begin{observation} \label{obs:bar}
Let $\mathcal{W}$ be a sequence of $n$ integers from zero to $n \choose 2$. For any arrangement $\mathcal{L}$ of $n$ lines, there exists $\mathcal{W}$ such that the largest side compatible subset $\mathcal{C}$ for $\mathcal{W}$ in the bar representation of $\mathcal{L}$ is of size at most $\frac{n}{2}$.
\end{observation}
\begin{proof}
Let $\mathcal{W}=\{0\}^\frac{n}{2}\{{n \choose 2}\}^\frac{n}{2}$. If $|\mathcal{C}|>\frac{n}{2}$, then there exits $c_1,c_2 \in \mathcal{C}$ such that $f(c_1)=0$ and $f(c_2)= {n \choose 2}$ (where $f$ is the mapping function for $\mathcal{C}$). However, this implies that the bar whose horizontal interval extends from $c_1$ to $c_2$ does not satisfy \ref{bar-P1} or \ref{bar-P2}.
\end{proof}

\begin{lemma}
Given any bar stack $\mathcal{B}_{n,n}$, there exists an $n$-subsequence $\mathcal{W}$ of $\langle 0~1~\cdots~n\rangle$, such that the largest side compatible subset for $\mathcal{W}$ in $\mathcal{B}_{n,n}$ is of size $n$.
\label{lem:all-wires}
\end{lemma}
\begin{proof}

Let $Y=\{ i+0.5 \mid i \in [0..n] \}$. Assume for each $y \in Y$, there is a wire at height $y$. We claim, for some $y$, there exists a side compatible marbling of size $n$ such that a marble associated with an endpoint of a bar is below the bar if and only if the bar is at height greater than $y$.

The endpoints of bars that lie on a pillar partition the pillar into a number of vertical intervals. Let $I^y_x$ denote the vertical interval of the pillar through $x$ that contains $y$, where $x \in [n]$ (See Figure~\ref{fig:interval-bar}). Let $X \subseteq [n]$ and $U^y_X=\bigcup\limits_{x\in X} I^y_x$. The union of two vertical intervals is simply the union of the $y$-values of the two intervals. Note that since $y \in I^y_x$ for all $x$, $U^y_{X}$ is also an interval. For an interval $I$, let $\sharp(I)$ denote the number of elements in $Y$ that are in $I$.
Let $E^y_X$ denote the set of endpoints of bars whose $x$-coordinates are in $X$ and $y$-coordinates are in $U^y_{X}$.
\begin{figure}[h]
\begin{center}
\begin{tikzpicture}[scale=.35]
	\foreach \i in {0,1,...,8}{
		\draw[blue,line width=1 pt] (0,\i+.5) -++ (8,0);
		\pgfmathparse{\i+.5}
		\node [blue,scale=.85] at (12,\pgfmathresult) {$y=\pgfmathresult$};
	}
	\draw[blue,dashed,line width=1 pt] (-3,5.5) -++ (13,0);		

	\draw[red, line width=1 pt] (7,1) -- (8,1);
	\draw[red,line width=1 pt] (2,2) -- (7,2);	
	\draw[red,line width=1 pt] (1,3) -- (8,3);		
	\draw[red,line width=1 pt] (1,4) -- (2,4);	
	\draw[red,line width=1 pt] (3,5) -- (4,5);
	\draw[red,line width=1 pt] (5,6) -- (6,6);	
	\draw[red,line width=1 pt] (4,7) -- (5,7);		
	\draw[red,line width=1 pt] (3,8) -- (6,8);	
	
\coordinate (P1) at (7,1);
\coordinate (P2) at (8,1);
\coordinate (P3) at (2,2);
\coordinate (P4) at (7,2);
\coordinate (P5) at (1,3);
\coordinate (P6) at (8,3);
\coordinate (P7) at (1,4); 
\coordinate (P8) at (2,4); 
\coordinate (P9) at (3,5);
\coordinate (P10) at (4,5);
\coordinate (P11) at (5,6);
\coordinate (P12) at (6,6);
\coordinate (P13) at (4,7);
\coordinate (P14) at (5,7);
\coordinate (P15) at (3,8);
\coordinate (P16) at (6,8);

\draw[gray,opacity=.4,line width=4pt] (P6) -- (8,8.5);
\draw[gray,opacity=.4,line width=4pt] (P4) -- (7,8.5);
\draw[gray,opacity=.4,line width=4pt] (P12) -- (6,0.5);
\draw[gray,opacity=.4,line width=4pt] (P11) -- (5,0.5);
\draw[gray,opacity=.4,line width=4pt] (P10) -- (P13);
\draw[gray,opacity=.4,line width=4pt] (P9) -- (P15);
\draw[gray,opacity=.4,line width=4pt] (P8) -- (2,8.5);
\draw[gray,opacity=.4,line width=4pt] (P7) -- (1,8.5);

\foreach \i in {1,...,16}	
	\draw[red,fill] (P\i) circle (3.5pt);

\end{tikzpicture}
\end{center}
\caption{Vertical intervals of pillars containing $y=5.5$.}
\label{fig:interval-bar}
\end{figure}
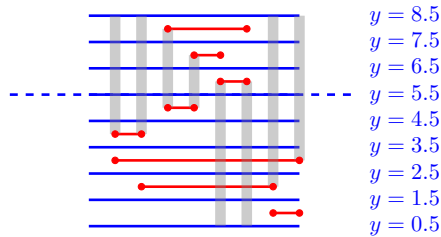

If (for some $y$) a marbling of size $n$ is such that every marble associated with an endpoint of a bar is below the bar if and only if the bar is at height greater than $y$, then for all $x  \in [n]$, the $y$-position of the marble on the pillar through $x$ is in $I^y_x$. Clearly, such a marbling satisfies~\ref{bar-P1} or~\ref{bar-P2}. In order to prove the claim, we need to show there exists such a marbling that is also valid (that is, no two marbles have the same $y$-position). It is easy to see that if for all $X \subseteq [n]$, $\sharp(U^y_X) \ge |X|$ then there is enough room so that each marble can have a distinct $y$-position. Therefore, if for some $y$, we have the property that for all $X \subseteq [n]$, $\sharp(U^y_X) \ge |X|$, then the claim follows.

Assume, for the sake of contradiction, that for all $y$, there exists $X_y \subseteq [n]$ such that $\sharp(U^y_{X_y}) < |X_y|$.

\begin{enumerate}[label=({\bf{Fact \arabic*})},leftmargin=2.8\parindent]
\item\label{Bprop1} $U^y_{X_y}$ is unbounded on at most one side.

If $U^y_{X_y}$ is unbounded on both sides, then $\sharp(U^y_{X_y})=n+1$. However, $|X_y| \le n$. Thus $\sharp(U^y_{X_y}) < |X_y|$ is contradicted.
\item \label{Bprop2} If $U^y_{X_y}$ is bounded from above, then at least $\lceil \frac{|X_y|}{2} \rceil -1$ elements of $Y \cap U^y_{X_y}$ are greater than $y$. If $U^y_{X_y}$ is bounded from below, then at least $\lceil \frac{|X_y|}{2} \rceil -1$ elements of $Y \cap U^y_{X_y}$ are less than $y$.

Let $U^y_{X_y}$ be bounded from above. Recall that $U^y_{X_y}=\cup_{x\in X_y} I^y_x$, and for all $x \in X_y$, $I^y_x$ contains $y$. This implies there are at least $|X_y|$ endpoints in $E^y_{X_y}$ whose heights are greater than $y$. Since no two bars lie at the same height, every integer height contains exactly two endpoints, and hence $E^y_{X_y}$ has at most two endpoints at every integer height. The vertical range between every two consecutive bars contains a distinct element of $Y$. Therefore, we infer that there are at least $\lceil \frac{|X_y|}{2} \rceil -1$ elements of $Y$ in $U^y_{X_y}$ that are greater than $y$. A similar argument works when $U^y_{X_y}$ is bounded from below.
\item \label{Bprop3}If $U^y_{X_y}$ is bounded on both sides, then
	\begin{enumerate}[label={(\bf{B\arabic*})},leftmargin=2.7\parindent]
		\item \label{Bp1}$|X_y|$ is even,
		\item \label{Bp2}$\sharp(U^y_{X_y}) = |X_y|-1$,
		\item \label{Bp3}$y$ is the (unique) median in $Y \cap U^y_{X_y}$,
		\item \label{Bp4} $E^y_{X_y}$ has exactly two endpoints at every integer height in $U^y_{X_y}$, and
		\item \label{Bp5} no two endpoints in $E^y_{X_y}$ that have different heights and are both above or below $y$ lie on the same pillar. 
	\end{enumerate}
~\ref{Bprop2} together with $|X_y| > \sharp(U^y_{X_y})$ implies $|X_y| > 2 \bigl\lceil \frac{|X_y|}{2} \bigr\rceil -1$.
Thus, $|X_y| = \bigl\lceil \frac{|X_y|}{2} \bigr\rceil +\bigl\lfloor \frac{|X_y|}{2} \bigr\rfloor \ge  2 \bigl\lceil \frac{|X_y|}{2} \bigr\rceil$, and hence $\bigl\lfloor \frac{|X_y|}{2} \bigr\rfloor \ge \bigl\lceil \frac{|X_y|}{2} \bigr\rceil$. This immediately implies~\ref{Bp1} and~\ref{Bp2}. Moreover, the number of elements of $Y$ in $U^y_{X_y}$ that are greater (or less) than $y$ is exactly $\frac{|X_y|}{2} -1$. This consequently implies~\ref{Bp3},~\ref{Bp4} and~\ref{Bp5}.
\item \label{Bprop4} If $U^y_{X_y}$ is bounded, then neither $U^{y-1}_{X_{y-1}}$ nor $U^{y+1}_{X_{y+1}}$ can be bounded.

Let $p$ be an endpoint whose height is below $y$ and above $y-1$ (i.e. at height $y-0.5$). Suppose $U^{y-1}_{X_{y-1}}$ is bounded. Note that~\ref{Bp4} implies that $p $ is in both $E^{y-1}_{X_{y-1}}$ and $E^{y}_{X_{y}}$ (since $p$ is immediately above or below $y-1$ and $y$). Thus, $E^{y-1}_{X_{y-1}}$ has an endpoint $p'$ that lies on the same pillar as $p$, and whose height is less than $y-1$. Recall that $U^y_{X_y}$ is bounded, and hence~\ref{Bp5} implies that $p' \notin E^{y}_{X_{y}}$. Thus $\sharp(U^{y-1}_{X_{y-1}}) \ge \sharp(U^y_{X_y})$. However, this implies that the pillar through $p$ contains two endpoints in $E^{y-1}_{X_{y-1}}$ that are both above $y-1$ and have different heights. This contradicts~\ref{Bp5}. A similar argument works for $U^{y+1}_{X_{y+1}}$.

 \end{enumerate}

We consider two cases.
\begin{enumerate}[label={\bf{Case \arabic*.}},leftmargin=2.7\parindent]
\item $n$ is even.\\
Let $y$ be the unique median value in $Y$ (i.e. $y=\frac{n}{2}+0.5$). By~\ref{Bprop1} we know $U^y_{X_y}$ is bounded on at least one side. If $U^y_{X_y}$ is unbounded, then~\ref{Bprop2} implies $\sharp(U^y_{X_y}) \ge \frac{n}{2}+1+\bigl\lceil \frac{|X_{y}|}{2} \bigr\rceil -1$; and since $|X_y| > \sharp(U^y_{X_y})$, we infer $|X_y| >n$, which is not possible. Therefore, $U^y_{X_y}$ is bounded (on both sides). Using~\ref{Bprop4}, we know both $U^{y-1}_{X_{y-1}}$ and $U^{y+1}_{X_{y+1}}$ are unbounded.

Consider $U^{y-1}_{X_{y-1}}$. Note that $|X_{y-1}| \le n$ and hence $\sharp(U^{y-1}_{X_{y-1}}) \le n-1$. If $U^{y-1}_{X_{y-1}}$ is unbounded from above, then $|X_{y-1}| > \sharp(U^{y-1}_{X_{y-1}}) \ge \frac{n}{2}+2+\bigl\lceil \frac{|X_{y-1}|}{2} \bigr\rceil -1$, and hence $|X_{y-1}| > n+2$, which is not possible. Therefore, $U^{y-1}_{X_{y-1}}$ is unbounded from below (and bounded from above). As a result, $|X_{y-1}| > \sharp(U^{y-1}_{X_{y-1}}) \ge \frac{n}{2}+\bigl\lceil \frac{|X_{y-1}|}{2} \bigr\rceil -1$, and hence $\bigl\lfloor \frac{|X_{y-1}|}{2} \bigr\rfloor > \frac{n-2}{2}$. Therefore (since $n$ is even) $|X_{y-1}|=n$ and $\sharp(U^{y-1}_{X_{y-1}})=n-1$.
Note that $|X_{y-1}|=n$ implies $E^{y-1}_{X_{y-1}}$ has exactly two endpoints at every integer height in $U^{y-1}_{X_{y-1}}$. Moreover, two endpoints in $E^{y-1}_{X_{y-1}}$ that are both above $y-1$ and have different heights need to lie on different pillars.

Similarly, we infer that $U^{y+1}_{X_{y+1}}$ is unbounded from above (and bounded from below); $|X_{y+1}|=n$ and $\sharp(U^{y+1}_{X_{y+1}})=n-1$; and consequently,
$E^{y+1}_{X_{y+1}}$ has exactly two endpoints at every integer height in $U^{y+1}_{X_{y+1}}$; and two endpoints in $E^{y+1}_{X_{y+1}}$ that are both below $y+1$ and have different heights need to lie on different pillars.

Let $F$ denote the four endpoints at heights $\frac{n}{2}$ and $\frac{n}{2}+1$ (i.e. the endpoints with heights immediately above or below $y$). No pillar going through an endpoint $p \in F$ can contain any other endpoint with height between $2$ and $n-1$. Recall that $U^y_{X_y}$ is bounded. This implies $\sharp(U^y_{X_y})=n-1$ and $|X_y|=n$. Moreover, $B_n$ and $B_{\frac{n}{2}}$ need to have identical horizontal intervals (and so do $B_1$ and $B_{\frac{n}{2}+1}$), which is not possible.

\item $n$ is odd.\\
Let $y$ be the smaller median value in $Y$ (i.e. $y=\frac{n-1}{2}+0.5$).
If $U^y_{X_y}$ is unbounded from above, then $|X_y| > \sharp(U^y_{X_y}) \ge \frac{n+1}{2}+1+\bigl\lceil \frac{|X_y|}{2} \bigr\rceil -1$. Hence $|X_y| > n+1$, which is not possible. If $U^y_{X_y}$ is unbounded from below, then $|X_y| > \sharp(U^y_{X_y}) \ge \frac{n+1}{2}+\bigl\lceil \frac{|X_y|}{2} \bigr\rceil -1$. Hence $2 \bigl\lfloor \frac{|X_y|}{2} \bigr\rfloor > n-1$, which implies $|X_y| > n$. But again, this is not possible. Therefore, $U^y_{X_y}$ is bounded.

Let $y'$ be the bigger median value in $Y$. Using a similar argument, we conclude that $U^{y'}_{X_{y'}}$ is also bounded. But this is a contradiction to \ref{Bprop4}.
\end{enumerate}
\end{proof}

\begin{lemma}
Given any bar stack $\mathcal{B}_{n,n}$, there exists an $\lfloor \frac{n}{2} \rfloor$-subsequence $\mathcal{W}$ of $\langle 0~1~\cdots~n\rangle$, such that the largest ordered side compatible subset for $\mathcal{W}$ in $\mathcal{B}_{n,n}$ is of size $\lfloor\frac{n}{2}\rfloor$.
\end{lemma}

\begin{proof}
Let $\mathcal{B}'$ be the set of bars in $\mathcal{B}_{n,n}$ with heights at most $\lfloor \frac{n}{2} \rfloor$.
We claim there exists an injective function $f:[n] \rightarrow [n]$ such that for any bar $B_i=[(a_i,i),(b_i,i)]$ in $\mathcal{B}'$, $f(a_i)>f(b_i) \ge i$ (where $i \in [\lfloor \frac{n}{2} \rfloor]$).
We can visualize this claim in terms of marbling: Let $Y=\{i +0.5 \mid i \in [0..n]\}$. Assume for each $y \in Y$, there is a wire at height $y$. We claim there exists a valid marbling of size $n$ such that for every bar $B_i \in \mathcal{B}'$, both marbles associated with the bar are above the bar; and moreover, the marble associated with the left endpoint of $B_i$ is higher than the marble associated with the other endpoint.

Note that if the claim is true, every bar $B_i \in \mathcal{B}'$ forms a partial order on the heights of marbles lying on pillars through $a_i$ and $b_i$; that is, it implies the inequalities $f(a_i)>f(b_i) \ge i$, where $f(x)+0.5$ is the height of the marble that is on pillar through $x$. We construct a directed graph representing all partial orders obtained from $\mathcal{B}'$, and prove that there exists a total order (consistent with all partial orders) on the heights of marbles such that each marble has a distinct height in $Y$.

We construct a directed 2-coloured graph $G=(V,E)$ as follows:
\begin{itemize}
\item  $V=V_p \cup V_h$, where $V_p = p_1 \cdots p_n$ and $V_h=h_1 \cdots h_{\lfloor \frac{n}{2} \rfloor}$ are sets of white and black vertices, respectively.
\item For every $i \in [\lfloor \frac{n}{2} \rfloor]$, there is a directed edge from $p_{a_i}$ to $p_{b_i}$, and another edge from $p_{b_i}$ to $h_i$.
\end{itemize}

Note that $G$ has no directed cycles. For every $v \in V_p$, let $r(v)$ denote the number of white vertices that are reachable from $v$. We refer to $r(v)$ as the \emph{r-value} of vertex $v$.
Decompose $G$ into weakly connected components . A \emph{weakly connected component} in a directed graph is a maximal connected component in the underlying undirected graph (that is, if replacing all directed edges with undirected edges).
Let
\begin{equation*}
H(C) = 
\begin{cases}
\max\{i \mid h_i \in V(C)\} & \text{if } |V(C)| >1\\
0 & \text{otherwise,}
\end{cases}
\end{equation*}
where $C$ is a weakly connected component and $V(C)$ is the set of vertices in $C$. We refer to $H(C)$ as the \emph{h-value} of $C$. The h-value of a component is zero if the component is \emph{trivial} (i.e. it consists of only one vertex). Note that a trivial component does not contain a black vertex. If $p_v$ is the only vertex in a component then the pillar through $v$ does not contain any endpoints of bars in $\mathcal{B}'$.

Let $c$ be the number of weakly connected components in $G$.
Let $G_1 G_2 \cdots, G_c$ be the ordering of the weakly connected components of $G$ in non-increasing order of their h-values. Let $P(G_i)$ denote the number of white vertices in $G_i$.
For each component $G_i$, sort its white vertices in non-increasing order of their r-values. 
Let $p_{\pi_i(0)}, p_{\pi_i(1)},\cdots$ be this ordering. We define function $f$ so that $f(\pi_i(j))= -j+n-\sum_{k=1}^{i-1} P(G_k)$.
Note that for every pair of white vertices $p_v$ and $p_u$, if $p_v$ can reach $p_u$, then $f(v) > f(u)$. Recall that for any white vertex $p_v \in G_i$, $n-\sum_{k=1}^{i}P(G_k)+1 \le f(v) \le n-\sum_{k=1}^{i-1}P(G_k)$.
In order to prove our claim, we need to show that $f$ is an injective function with range $[n]$.
This immediately follows if for every component $G_i$, $H(G_i) \le n-\sum_{k=1}^{i}P(G_k)+1$.

Note that for every $i \in [c]$,  $0 \le H(G_i) \le \lfloor \frac{n}{2} \rfloor$.
If $H(G_i)=0$, then $H(G_i) \le n-\sum_{k=1}^{i}P(G_k)+1$ because $\sum_{k=1}^c P(G_k) \le n$.  Recall that $H(G_i)=0$ if and only if $|V(G_i)|=1$. The h-values of the components that contain more than one vertex are all distinct. Therefore, if $H(G_i) >0$, then $H(G_i)\le \lfloor \frac{n}{2} \rfloor - (i-1)$.
Let $E^i_p$ denote the number of edges in the subgraph induced by the white vertices in $G_i$. We know $P(G_i) \le E^i_p +1$, and hence $\sum_{k=1}^i P(G_k) \le i+ \sum_{k=1}^i E^k_p$. Note that $\sum_{k=1}^c E^k_p \le \lfloor \frac{n}{2} \rfloor$. Therefore, $n-\sum_{k=1}^{i}P(G_k)+1 \ge \lceil \frac{n}{2} \rceil -i +1$, and subsequently $n-\sum_{k=1}^{i}P(G_k)+1 \ge H(G_i)$. This concludes the proof of our claim.

In the following, we use the function $f$ defined in our claim to prove the lemma. Initialize $\mathcal{C}=[n]$. Using our claim, we know that for every bar $B_i=[(a_i,i),(b_i,i)]$ in $\mathcal{B}_{n,n}$, where $i \in [\lfloor \frac{n}{2} \rfloor]$, $f(a_i)>f(b_i) \ge i$. For all $j \in [\lfloor \frac{n}{2} \rfloor +1 .. n]$ such that neither $f(a_j)>f(b_j) \ge j$ nor $f(a_j)<f(b_j) < j$ holds, we remove $a_j$ or $b_j$ from $\mathcal{C}$. It is easy to verify that at the end $|\mathcal{C}| \ge n -\lceil \frac{n}{2} \rceil = \lfloor \frac{n}{2} \rfloor$.  Using function $f$ on $\mathcal{C}$, we guarantee that $\mathcal{C}$ is an ordered side compatible subset for
a subsequence of size $|\mathcal{C}| \ge \lfloor \frac{n}{2} \rfloor$ obtained from 
$\langle 1~\cdots~n\rangle$.
\end{proof}

\begin{observation}
If for any $n$, any bar stack $\mathcal{B}_{n,n}$, and any sequence $\mathcal{W}_1$ of $2n+2$ integers, there is a side compatible subset of size $n$ for $\mathcal{W}_1$ in $\mathcal{B}_{n,n}$, then for any sequence $\mathcal{W}_2$ of $n$ integers, there is a side compatible subset of size $\lfloor\frac{n}{2}\rfloor-2$ for $\mathcal{W}_2$ in $\mathcal{B}_{n,n}$.

\label{obs:2nTOn}
\end{observation}
\begin{proof}
For any bar stack $\mathcal{B}_{n,n}$, there exists a subset $\mathcal{C} \subseteq [n]$ of size $\lfloor\frac{n}{2}\rfloor$ such that the number of bars in $\mathcal{B}_{n,n}$ whose endpoints' $x$-coordinates both belong to $\mathcal{C}$ is at most $\lfloor\frac{n}{2}\rfloor$ (note that for any $\lfloor\frac{n}{2}\rfloor$-subset $C$, either $C$ or a $\lfloor\frac{n}{2}\rfloor$-subset from $[n]\setminus C$ satisfies this property). Let $\mathcal{B}'$ denote the set of bars in $\mathcal{B}_{n,n}$ induced by $\mathcal{C}$ ($\mathcal{B}'$ is a set of at most $\lfloor\frac{n}{2}\rfloor$ bars whose heights range from $1$ to $n$). If the premise of the statement in Observation~\ref{obs:2nTOn} is true, then $\mathcal{C}$ is a side compatible subset for any sequence of $2\lfloor \frac{n}{2}\rfloor+2$ integers in $\mathcal{B}'$ (here a number $\mathit{w}$ in the sequence corresponds to a wire that is below any bar whose height is greater than $\mathit{w}$ and is above any bar whose height is equal or less than $\mathit{w}$).
As a result, there is a side compatible subset of size at least $\lfloor\frac{n}{2}\rfloor-2$ for any sequence of $n$ integers in $\mathcal{B}_{n,n}$. This concludes the proof.
\end{proof}

\begin{claim}
Let $\mathcal{W}$ be any sequence of $2n+2$ integers from zero to $n$. 
Given any bar stack $\mathcal{B}_{n,n}$, there exits a side compatible subset $\mathcal{C}$ for $\mathcal{W}$ that is of size $\frac{n}{2}$.
\label{lem:n-bar-n-wire}
\end{claim}

\begin{proof}[Proof sketch]
We use the notation used in the proof of Lemma~\ref{lem:all-wires}.
Let $X_y$ be the smallest subset for which $\sharp(U^y_{X_y})<|X_y|$.
First note that $U^y_{X_y}$ is bounded for $y \in \{n-0.5, n+0.5, n+1.5, n+2.5\}$ (i.e. any of the four middle wires).
There cannot exist more than $n$ consecutive wires whose $U^y_{X_y}$ is bounded.
Thus, there exists $y$ such that $U^y_{X_y}$ in unbounded and $\sharp(U^y_{X_y}) \ge \frac{n}{2}+1$. This implies that there exists a side compatible subset of size $\frac{n}{2}$.
\end{proof}

\begin{lemma}
There exist a bar stack $\mathcal{B}_{{n \choose 2},n}$ and a sequence $\mathcal{W}_{l,{n \choose 2}}$, with $l = \Theta(n^2)$, for which the largest side compatible subset is less than $n$.
\end{lemma}
\begin{proof}
First note that we can partition the edges of $K_n$ into $\frac{n}{2}$ paths of size $n-1$. See Figure~\ref{fig:kn-paths}. We construct a bar stack using these paths.
Let $P_0,\dots, P_{\frac{n}{2}-1}$ denote the set of these paths.
Assume that the vertices of $K_n$ are labeled $1,\dots,n$.
Let $e^i_k = (u^i_k,v^i_k)$ denote the $k$-th edge in $P_i$.
$e^i_k$ corresponds to $B_{i*(n-1)+k}$, which is a segment from $(u^i_k,i*(n-1)+k )$ to $(v^i_k,i*(n-1)+k )$.
Let $\mathcal{B}_i$ denote the set of bars $\{B_{i*(n-1)+k} \mid k \in [n-1]\}$. We refer to $\mathcal{B}_i$ as a \emph{block} (of segments).
Recall that $\mathcal{W}_{l,{n \choose 2}}$ corresponds to a set of wires.
Let $l=(\frac{n}{2}+1)(n-1)$.
We construct $\mathcal{W}_{l,{n \choose 2}}$ so that there are $n-1$ wires between every two consecutive blocks (i.e. there are $n-1$ wires between $B_{i*(n-1)+n-1}$ and $B_{(i+1)*(n-1)+1}$ for all $i$ in $\{0, \dots,\frac{n}{2}-2\}$); $n-1$ wires below $\mathcal{B}_0$; and $n-1$ wires above $\mathcal{B}_{\frac{n}{2}-1}$.
Note that there is no wire between the segments that belong to the same block.
A set of wires $W$ \emph{encompasses} a block if the block is between two wires in $W$.
Any set of $n$ wires encompasses a block, and each block of segments is incident to all $n$ pillars.
Thus, the size of the largest side compatible subset is at most $n-1$.
\begin{figure}[h]
\begin{center}
\begin{tikzpicture} [scale=.7]
	\def \r {1};
	
	\foreach \j in {0,1,...,4}{
		\coordinate (O\j) at (3*\j,0);
		\foreach \i in {0,1,...,10}{
			\path (O\j) -- +(36*\i: \r cm) coordinate (D\i\j);	
			\draw[fill] (D\i\j) circle (2pt);
		}
	}
	
	\draw (D00) -- (D10) -- (D90) -- (D20) -- (D80) -- (D30) --( D70) -- (D40) -- (D60) -- (D50);
	\draw (D11) -- (D21) -- (D01) -- (D31) -- (D91) -- (D41) --( D81) -- (D51) -- (D71) -- (D61);
	\draw (D22) -- (D32) -- (D12) -- (D42) -- (D02) -- (D52) --( D92) -- (D62) -- (D82) -- (D72);
	\draw (D33) -- (D43) -- (D23) -- (D53) -- (D13) -- (D63) --( D03) -- (D73) -- (D93) -- (D83);
	\draw (D44) -- (D54) -- (D34) -- (D64) -- (D24) -- (D74) --( D14) -- (D84) -- (D04) -- (D94);
		
\end{tikzpicture}
\end{center}
\caption{Partitioning the edges of $K_n$ into $\frac{n}{2}$ paths of size $n-1$.}
\label{fig:kn-paths}
\end{figure}
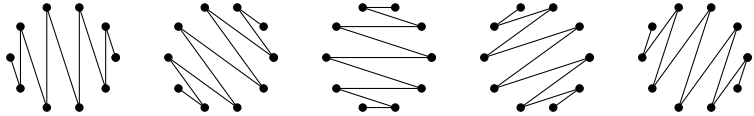
\end{proof}


\begin{question}
Does there exist a bar stack $\mathcal{B}_{l,n}$ together with a sequence $\mathcal{W}_{n,l}$, where $l=\Theta(n^2)$, such that the largest side compatible subset for $\mathcal{W}_{n,l}$ in $\mathcal{B}_{l,n}$ is of size $o(n)$?
\end{question}

Given a bar stack $\mathcal{B}$, let $G(\mathcal{B})$ denote the geometric graph whose vertices are the endpoints of segments in $\mathcal{B}$ such that two vertices $u$ and $v$ are connected by an edge if $(u,v)$ is a segment in $\mathcal{B}$ or if $u$ and $v$ belong to the same pillar. A \emph{simple} cycle in a graph is a cycle that does not self-intersect.
\begin{observation}
Let $\mathcal{B}$ be the bar stack obtained from the bar representation of a line arrangement. $G(\mathcal{B})$ does not admit any simple cycles.
\end{observation}

\begin{question}
Let $\mathcal{W}=\mathit{w}_1 \mathit{w}_2 \dots \mathit{w}_n$ be any sequence of $n$ integers from zero to ${n \choose 2}$. Let $\mathcal{B}$ be any bar representation of a line arrangement.
Is the largest side compatible subset for $\mathcal{W}$ in $\mathcal{B}$ linear-sized?
\label{lem:bar-rep}
\end{question}

\section{Generalizations of crossing families}
In this section, we study two (geometric) generalizations of crossing families.

The first generalization may be viewed as a generalization of spoke sets in the dual plane. In Section~\ref{sec:ss}, we summarize the definitions and results on spoke sets in the primal and dual planes.
We then introduce a more generalized notion for the dual of spoke sets, which we call ``M-semialternating paths''.
M-semialternating paths are related to pseudolines in two-coloured line arrangements that intersect (in order) lines of alternate colours in the arrangement.
We give a few upper bounds on the sizes of certain M-semialternating paths.
As a step towards showing that spoke sets are linear-sized, we prove that the sizes of certain M-semialternating paths are connected. This also helps us improve the upper bound on the size of spoke sets from $\frac{9n}{20}$ to $\frac{n}{4}+1$. 

The second generalization generalizes a crossing family to a family of segments such that for every pair of segments, the supporting line of one intersects the interior of the other. We call such a family of segments a \emph{stabbing family}. We show that the size of the largest stabbing family in an $n$-point set is $\frac{n}{2}$.

\subsection{Spoke sets} \label{sec:ss}
\citet{bose} studied partitions of complete geometric graphs into plane trees and introduced the notion of spoke sets, which is closely related to crossing families.

\begin{definition}
\label{def:ss}
Let $\mathcal{P}$ be a set of points in general position.  A set $\mathcal{L}$ of pairwise non-parallel lines such that  each open unbounded region in the arrangement of $\mathcal{L}$ has at least one point of $\mathcal{P}$ is called a \emph{spoke set} for $\mathcal{P}$. The size of a spoke set $\mathcal{L}$ is the number of lines in $\mathcal{L}$.
\end{definition}

Note that if we extend the segments of a crossing family to lines and
perturb them infinitesimally with a clockwise rotation
so that the endpoints of each segment are on different sides of its perturbed line, the resulting line arrangement has exactly one point in every unbounded region and hence forms a spoke set.
While any crossing family of size $k$ guarantees a spoke set of size $k$, the reverse is not true. It is not known whether the order of magnitude of the sizes of crossing families and spoke sets are the same or not.

\citet{gcrf17} studied
spoke sets in the dual plane and introduced the notion of spoke paths:

\begin{definition} 
\label{def:sp}
A \emph{cell-path} in an arrangement $\mathcal{L}$ of lines (or pseudolines) is a sequence of cells in the arrangement such that consecutive cells share an edge. The \emph{length} of a cell-path is one less than the number of cells involved.
A cell-path $\mathcal{C}=\langle C_0,C_1,\dots, C_k\rangle$ is \emph{AB-semialternating} if for every even $i<k-1$, $C_i$ is above the line separating $C_i$ and $C_{i+1}$ if and only if $C_{i+2}$ is above the line separating $C_{i+1}$ and $C_{i+2}$. A cell-path is \emph{line-monotone} if the lines extending the edges shared by two consecutive cells are all distinct.
A cell-path in a subarrangement of $\mathcal{L}$ that is line-monotone and AB-semialternating is called a \emph{spoke path} for $\mathcal{L}$.
A cell-path is \emph{admitted} by the set of lines that extend the edges shared by two consecutive cells in the cell-path.
\end{definition}

Let $P$ be a point set and $P^\star$ denote the dual of $P$. Let $\mathcal{C}=\langle C_0,C_1,\dots, C_{2k}\rangle$ be a spoke path for $P^\star$. $\mathcal{L}=\{p^\star_i \mid p_i \in C_{2i}, i \in[k]\}$ is a spoke set for $P$ (where $p^\star_i$ is the dual of $p_i$).

\begin{lemma} [\citet{gcrf17}, Lemma 1]
Let $\mathcal{P}$ be a point set and $\mathcal{P^{\star}}$ be the dual line arrangement for $\mathcal{P}$. $\mathcal{P}$ contains a spoke set of size $k$ if and only if $\mathcal{P^\star}$ contains a spoke path of length $2k$.
\end{lemma}

\begin{conjecture} [\citet{gcrf17}, Lemma 2]
Let $\mathcal{P}=\mathcal{P}_{\mathcal{B} \vdash \mathcal{R}}$ be a $1$-avoiding $2k$-point set. The dual line arrangement $\mathcal{P^\star}$ contains a spoke path of length $k+2$, if $k$ is even, and $k+3$, if $k$ is odd.
\label{schnider-claim}
\end{conjecture}

Conjecture~\ref{schnider-claim} is claimed to be proved by~\citet{gcrf17} but the proof given does not seem to be correct.
The sketch of their proof is the following: Given the line arrangement $\mathcal{P}^\star=\mathcal{B}^\star \cup \mathcal{R}^\star$, they construct
an ``extended diagram'' for $\mathcal{P}^\star$, which is a horizontal wiring diagram for $\mathcal{R}^\star$ and a vertical wiring diagram for $\mathcal{B}^\star$,
with the same intersection order along every pseudoline as in $\mathcal{P}^\star$. A \emph{wiring diagram} is an arrangement of pseudolines consisting of piecewise linear ``wires'', where the wires (i.e. pseudolines) are horizontal except for small neighbourhoods of their crossings with other wires. They change the extended diagram through a number of steps to reach a certain type of an extended diagram, which has two spoke paths of size $2k$. They then reverse their moves to get back to the initial extended diagram, and in each step modify those spoke paths accordingly so that they become spoke paths of the new extended diagram obtained. Note that the size of the spoke paths may shrink at each step. They show that at least one of those spoke paths is large enough when they return to the initial diagram. However, the problem with their proof is that their proposed rules for getting new diagrams from old ones may cause the pseudolines to double cross and hence what is obtained is not guaranteed to be an extended diagram. It is easy to change the rules so that at each step we can guarantee that no two pseudolines double cross, however, with the new rules, the process of modifying the spoke paths when reversing the moves becomes problematic. That is, either it is not easy to maintain a spoke path or to guarantee having a large enough one.

\subsection{M-semialternating paths}
\label{sec:m}
Here, we generalize the notion of spoke paths. We start by some terminologies.

We call a polygonal chain whose line segments connect points in consecutive cells of a spoke path admitted by $\mathcal{A} \subseteq \mathcal{L}$ an \emph{AB-semialternating path} for $\mathcal{A}$. 
The definition of spoke paths (Definition~\ref{def:sp}) implies that any AB-semialternating path for $\mathcal{A}$ starts and ends in median cells of $\mathcal{A}$.
A \emph{median cell} is a cell of the arrangement that has an equal number of lines above and below it. A cell is of \emph{level} $k$ if each point in its interior is above exactly $k$ lines.
Figure~\ref{fig:sp-m} shows an AB-semialternating path for a subarrangement.

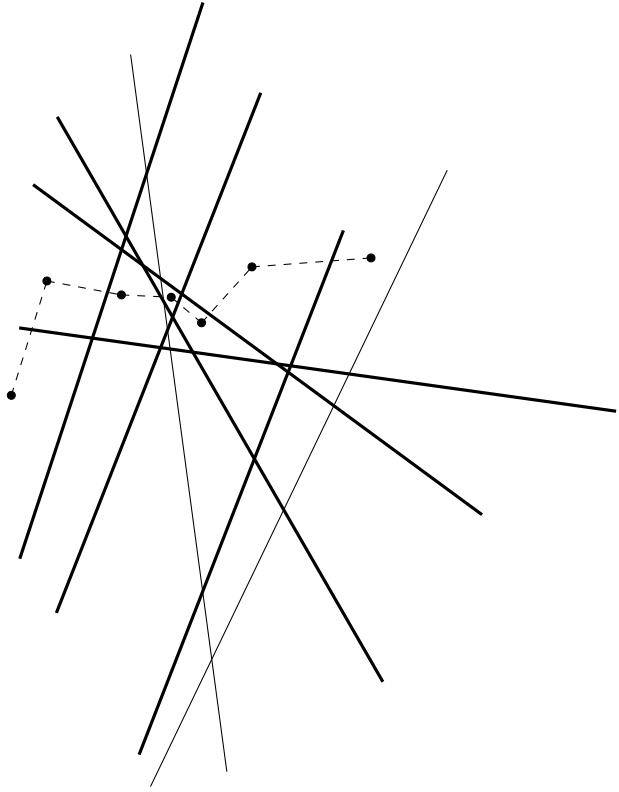
\begin{figure}[h]
	\begin{center}
\begin{tikzpicture}[scale=1.5]
\begin{scope}[rotate=35]
	\draw[very thick] (0.09,0.28)--(4,-3.36);
	\draw[very thick] (0.92,1.25)--(2.5,-3.43);
	\draw[very thick] (1.44,1.62)--(0.93,-4.14);
	\draw[] (2.29,1.7)--(-0.66,-4);
	
	\draw[very thick] (-1.08,-1.4)--(3.08,1.71);
	\draw[very thick] (-1.09,-1.98)--(3.04,0.76);
	\draw[very thick] (-1.21,-3.43)--(2.94,-0.66);
	\draw[] (-1.29,-3.72)--(4,-0.75);
	
	\coordinate (A) at (-0.31,-0.17);
	\coordinate (B) at (0.53,0.48);
	\coordinate (C) at (1,0);
	\coordinate (D) at (1.35,-0.27);
	\coordinate (E) at (1.44,-0.61);
	\coordinate (F) at (2.09,-0.46);			
	\coordinate (G) at (3,-1);			
	
	\foreach \p in {A,B,...,G}
			\draw[fill] (\p) circle (1pt);	
			
	\draw[dashed] (A) --(B);			
	\draw[dashed] (B) --(C);			
	\draw[dashed] (C) --(D);			
	\draw[dashed] (D) --(E);			
	\draw[dashed] (E) --(F);			
	\draw[dashed] (F) --(G);			
					
\end{scope}
	
\end{tikzpicture}
	\end{center}
	\caption{The dashed polygonal chain is an AB-semialternating path for the arrangement marked in bold. No spoke path of size eight exists for the line arrangement containing all lines.}
	\label{fig:sp-m}
\end{figure}

Let $\mathcal{C}=\langle C_0, C_1,\cdots,C_{|\mathcal{A}|}\rangle$ be a line-monotone cell-path admitted by $\mathcal{A}$, where $|\mathcal{A}|$ is even. An \emph{even} cell in $\mathcal{C}$ is a cell whose level has the same parity as that of $C_0$.
Let the \emph{level-sequence} of $\mathcal{C}$ be $seq(\mathcal{C})=\langle s_0,s_1,\cdots,s_{\frac{|\mathcal{A}|}{2}}\rangle$, where $s_i$ is the level of $C_{2i}$ in $\mathcal{C}$.
Note that every even cell in a spoke path admitted by $\mathcal{A}$ is a median cell of $\mathcal{A}$; hence the level-sequence of a spoke path admitted by $\mathcal{A}$ is $\{\frac{\mathcal{A}}{2}\}^{\frac{\mathcal{A}}{2}+1}$.

We generalize the notion of an AB-semialternating path admitted by a subarrangement $\mathcal{A}$ so that it intersects each line in $\mathcal{A}$ once and the level-sequence of the underlying cell-path is non-decreasing.
A pair of unbounded regions $U_1$ and $U_2$ in a line arrangement are \emph{antipodal} if for every line of the arrangement, the points in $U_1$ and $U_2$ are on opposite sides.
The fact that the underlying cell-path intersects each line in $\mathcal{A}$ exactly once, implies that our generalized path needs to start and end in antipodal unbounded regions (of $\mathcal{A}$).

\begin{definition}
Let $\mathcal{L}=\mathcal{L_B} \cup \mathcal{L_R}$ be a two-coloured line arrangement, where $|\mathcal{L}_\mathcal{B}|=|\mathcal{L}_\mathcal{R}|$. 
Let $\ell$ be a pseudoline intersecting each line in $\mathcal{L}$ once. Let $\ell(\mathcal{L})=l_1, l_2,\dots, l_{|\mathcal{L}|}$ define the order in which $\ell$ intersects the lines of $\mathcal{L}$.
We say $\ell$ is monotonically semialternating, or M-semialternating for short, for $\mathcal{L}$ if 
\begin{enumerate}[label={(\bf{M\arabic*)}},leftmargin=2.5\parindent]
\item \label{semialt}for every odd $i<|\mathcal{L}|$, $l_i$ and $l_{i+1}$ are of different colours, and 
\item the level-sequence of the underlying cell-path for $\ell$ is non-decreasing.
\end{enumerate}
 If $\mathcal{A} \subseteq \mathcal{L}$ admits an M-semialternating pseudoline, we say $\mathcal{L}$ contains an \emph{M-semialternating path} of \emph{size} $|\mathcal{A}|$. We say $\ell$ is \emph{semialternating} if it satisfies \ref{semialt}.
\end{definition}

Let $\mathcal{P}_{\mathcal{B} | \mathcal{R}}$ be a point set $\mathcal{B} \cup \mathcal{R}$, where $\mathcal{B}$ is a blue point set and $\mathcal{R}$ is a red point set, such that $\mathcal{B}$ and $\mathcal{R}$ are two equal-sized sets of points that are separable by a vertical line. We call $\mathcal{P}_{\mathcal{B} | \mathcal{R}}$ a \emph{color-separable} point set. Let $\mathcal{L}_\mathcal{B}$ and $\mathcal{L}_\mathcal{R}$ be the dual line arrangements for $\mathcal{B}$ and $\mathcal{R}$, respectively. Let $\mathcal{L}=\mathcal{L}_\mathcal{B} \cup \mathcal{L}_\mathcal{R}$, where $\mathcal{L}_\mathcal{B}$ is blue and $\mathcal{L}_\mathcal{R}$ is red. We call $\mathcal{L}$ a \emph{color-separable} line arrangement. Recall that a line arrangement is \emph{$1$-avoiding} if its dual point configuration is $1$-avoiding.


We show that the sizes of certain M-semialternating paths in color-separable line arrangements are connected. M-semialternating paths that start in different levels of a color-separable line arrangement correspond to different concepts in the dual plane.
We
exploit this correspondence
to improve the upper bound on the size of spoke sets.
We also give a linear upper bound on the size of subarrangements admitting semialternating pseudolines. We show that if we consider lines rather than pseudolines, there exist $1$-avoiding line arrangements whose largest subarrangement admitting an M-semialternating line is of constant size.

\begin{definition}
For a point set $P$ and a line $l$, let $A(P,l)$ denote the set of points in $P$ that are above $l$.
A \emph{parallel set} of size $k$ in a two-coloured point set $\mathcal{P}=\mathcal{P}_{\mathcal{B} | \mathcal{R}}$ is a set of lines $L={L_0,L_1,\dots,L_k}$
for which there exist $B \subseteq \mathcal{B}$ and $R \subseteq \mathcal{R}$, where $|B|=|R|=k$, such that for any $1 \le i \le k$, $A(B,L_{i-1}) \subsetneq A(B,L_i)$ and $A(R,L_{i-1}) \subsetneq A(R,L_i)$.
A \emph{focal parallel set} is a parallel set whose lines all intersect at the same point.
\end{definition}

\begin{observation}
Let $\mathcal{P}=\mathcal{P}_{\mathcal{B} | \mathcal{R}}$ be a color-separable point set.
Let $\mathcal{P}^\star$ be the dual line arrangement for $\mathcal{P}$. $\mathcal{P}$ contains a parallel set (or focal parallel set) of size $k$ if and only if a $2k$-subset $\mathcal{L}$ of $\mathcal{P}^\star$ admits an M-semialternating pseudoline (or line) that starts and ends in cells of $\mathcal{L}$ that are of levels zero and $2k$.
\end{observation}

\begin{lemma}
Let $\mathcal{P}=\mathcal{P}_{\mathcal{B}|\mathcal{R}}$ be a color-separable $2n$-point set, and let $\mathcal{L}$ be the dual line arrangement for $\mathcal{P}$. If $\mathcal{L}$ contains an M-semialternating path of size $2k$, then either the spoke set or the parallel set for $\mathcal{P}$ is of size at least $\frac{k}{2}$.
\label{lem:ssORps}
\end{lemma}
\begin{proof}
Let $\ell$ denote an M-semialternating pseudoline for $\mathcal{L}' \subseteq \mathcal{L}$, where $|\mathcal{L}'|=2k$.
Let $\mathcal{C}$ denote the underlying cell-path for $\ell$ in $\mathcal{L}'$.
Let $\mathcal{S}= \langle s_0,s_1,\cdots,s_k \rangle$ denote the level-sequence of $\mathcal{C}$. Note that $s_0+s_k=2k$ and for every $i \in [k]$, $s_i-s_{i-1} \in \{0,2\}$.
Assume $s_0=\mathtt{x}$ (thus $s_k=2k-\mathtt{x}$), where $0 \le \mathtt{x} \le k$. $\mathcal{S}$ contains $\frac{2k-\mathtt{x}-\mathtt{x}}{2}+1=k-\mathtt{x}+1$ distinct numbers and $\mathtt{x}$ repetitions.
Let $\mathcal{I}_1= \{ i \mid  s_i-s_{i-1}=2\}$ and $\mathcal{I}_2=[k] \setminus \mathcal{I}_1$. Note that $|\mathcal{I}_1|=k-\mathtt{x}$ and $|\mathcal{I}_2|=\mathtt{x}$. Let $B(i)$ denote the set of lines that are below the $(i+1)$-th even cell in $\mathcal{C}$ (note that $|B(i)|=s_i$). Let $\mathcal{L}_1=\bigcup \limits_{i \in \mathcal{I}_1} B(i) \setminus B(i-1)$, and $\mathcal{L}_2=\bigcup \limits_{i \in \mathcal{I}_2} B(i) \triangle B(i-1)$, where $\triangle$ represents the symmetric difference. Recall that $\ell$ intersects each line in $\mathcal{L}'$ exactly once.
$\ell$ forms an M-semialternating pseudoline for $\mathcal{L}_1$ that starts and ends in cells of levels zero and $|\mathcal{L}_1|=2|\mathcal{I}_1|=2(k-\mathtt{x})$ (Note that the level-sequence of the underlying cell-path for $\ell$ in $\mathcal{L}_1$ is strictly increasing). Similarly, $\ell$ forms a spoke path of size $|\mathcal{L}_2|=2|\mathcal{I}_2|=2\mathtt{x}$ for $\mathcal{L}_2$ (The level-sequence of the underlying cell-path for $\ell$ in $\mathcal{L}_2$ is a constant sequence). Hence the dual point set for $\mathcal{L}_1$ has a parallel set of size $k-\mathtt{x}$ and the dual point set for $\mathcal{L}_2$ has a spoke set of size $\mathtt{x}$. This consequently implies that either the parallel set or the spoke set for $\mathcal{P}$ is of size at least $\frac{k}{2}$. 
\end{proof}

Let $\ell$ be an M-semialternating pseudoline for the color-separable line arrangement $\mathcal{L}$, where $|\mathcal{L}|=2k$. Let $\mathcal{C}=\langle C_0,C_1,\cdots,C_{2k} \rangle$ denote the underlying cell-path for $\ell$ in $\mathcal{L}$.
Let $o$ be a point in $C_0 \cap \ell$.
Let $Rot_\alpha(\cdot)$ be the function that rotates the input $\alpha$ degrees clockwise about $o$.
Clearly, for any $\alpha$, $Rot_\alpha(\ell)$ remains semialternating for $Rot_\alpha(\mathcal{L})$. However, the level-sequence of the underlying cell-path for $Rot_{\alpha}(\ell)$ in $Rot_{\alpha}(\mathcal{L})$ may become non-monotone. Moreover, the lines below $Rot_\alpha(C_0)$ are not (necessarily) of the same color.
It is easy to see (proof below) that given any M-semialternating pseudoline, there exists an angle $\alpha$ such that the level-sequence of the underlying cell-path for $Rot_\alpha(\ell)$ in $Rot_\alpha(\mathcal{L})$ is strictly increasing. However, there may not exist an angle $\alpha$ such that the level-sequence of the underlying cell-path for $Rot_\alpha(\ell)$ in $Rot_\alpha(\mathcal{L})$ is a constant sequence.

Let $\mathcal{S}=\langle s_0,s_1,\cdots, s_k \rangle$ denote the level-sequence of $\mathcal{C}$. $s_i-s_{i-1} \in \{0,2\}$. Since $\ell$ crosses each line of $\mathcal{L}$ once and the level-sequence is non-decreasing, we infer that $0 \le s_0 \le k$; thus, since $\mathcal{L}$ is color-separable, all the lines below $C_0$ are of the same color. (Recall that a vertical line to the left of all intersection points of $\mathcal{L}$ intersects all the red and all the blue lines consecutively.)
Let $\theta^-$ be such that the set of lines below $o$ in $Rot_{\theta^-}(\mathcal{L})$ is the same as the set below $o$ in $\mathcal{L}$ except one less.
($\theta^-$ is undefined if the level of $o$ in $\mathcal{L}$ is zero.)
Similarly, let $\theta^+$ be such that the set of lines below $o$ in $Rot_{\theta^+}(\mathcal{L})$ is the same as the set below $o$ in $\mathcal{L}$ except one more.
Let $\overline{\mathcal{S}}=\langle \overline{s}_0,\overline{s}_1,\cdots,\overline{s}_k \rangle$ and $\overset{+}{\mathcal{S}}=\langle \overset{+}{s}_0,\overset{+}{s}_1,\cdots,\overset{+}{s}_k \rangle$ be the level-sequences of the underlying cell-paths for $Rot_{\theta^-}(\ell)$ in $Rot_{\theta^-}(\mathcal{L})$ and $Rot_{\theta^+}(\ell)$ in $Rot_{\theta^+}(\mathcal{L})$, respectively.
There exists $d \in[k]$ such that for all $j<d$, $\overset{+}s_j=s_j+1$ and for all $j \ge d$, $\overset{+}s_j=s_j-1$. Similarly, there exists $u \in[k]$ such that for all $j<u$, $\overline{s}_j=s_j-1$ and for all $j \ge u$, $\overline{s}_j=s_j+1$.

$\overline{\mathcal{S}}$ is monotone. The difference between two consecutive elements in $\overline{\mathcal{S}}$ can either be zero or two (since $\overline{\mathcal{S}}$ represents the level-sequence of a cell-path). This implies that $u$ is such that $s_{u}=s_{u-1}$, and hence $\overline{\mathcal{S}}$ contains more distinct numbers compared to $\mathcal{S}$. The lines below $o$ in $\mathcal{L}$ are all of the same color, and the lines below $o$ in the rotated arrangement (with rotation angle $\theta^-$) are still of the same color. Therefore, we can continue rotating by $\theta^-$ degrees until the (rotated) pseudoline starts from the zero level and we get a strictly increasing level-sequence.

$\overset{+}{\mathcal{S}}$ is monotone only if $s_d \neq s_{d-1}$. Starting with the pair $\mathcal{L}$ and $\ell$ and rotating them  $\theta^+$ degrees, $Rot_{\theta^+}(\ell)$ remains M-semialternating with respect to $Rot_{\theta^+}(\mathcal{L})$. However, if we continue rotating this way, the level-sequence of the underlying cell-path may become non-monotone. Hence, by rotating an M-semialternating pseudoline (together with the arrangement),  it may not be possible to get a level-sequence that is constant.

\begin{lemma}
The minimum size of the largest spoke set among all configurations of color-separable $n$-point sets is the same as the minimum size of the largest parallel set in a color-separable $n$-point set. The statement still holds if we replace ``color-separable'' with ``$1$-avoiding''.
\label{lem:pset-ss}
\end{lemma}

\begin{proof}
Let $\mathcal{P}=\mathcal{P}_{\mathcal{B} | \mathcal{R}}$ be a color-separable point set. Let $ss(\mathcal{P})$ and $ps(\mathcal{P})$ denote the sizes of the maximum spoke set and parallel set for $\mathcal{P}$, respectively.
We show that for every color-separable point set $\mathcal{P}$ there exists a color-separable point set $\mathcal{P}'$ such that $ss(\mathcal{P'})=ps(\mathcal{P})$ and $ps(\mathcal{P}')=ss(\mathcal{P})$. We construct $\mathcal{P}'$ as follows:
\begin{itemize}
\item Assume by rotation and translation that $\mathcal{B}$ and $\mathcal{R}$ lie on the right and left sides of the $y$-axis, respectively.
\item Let $\mathcal{L}$ be the dual line arrangement for ${\mathcal{P}}$.
\item Let $\mathcal{L}'$ be the line arrangement that is obtained by rotating $\mathcal{L}$ $90$ degrees clockwise.
\item $\mathcal{P}'$ is the dual point configuration for $\mathcal{L}'$.
\end{itemize}
We know there exists $L_s \subseteq \mathcal{L}$, where $|L_s|=2 \cdot ss(\mathcal{P})$, such that $L_s$ admits an M-semialternating pseudoline $\ell_s$ that starts and ends in median cells of $L_s$. Similarly, there exists $L_p \subseteq \mathcal{L}$, where $|L_p|=2 \cdot ps(\mathcal{P})$, such that $L_p$ admits an M-semialternating pseudoline $\ell_p$ that starts and ends in cells of $L_p$ that are of levels zero and $|L_p|$.
Let $Rot(\cdot)$ denote the function that rotates the input $90$ degrees clockwise about the origin. Recall that $\mathcal{L}'=Rot(\mathcal{L})$.
Clearly, $Rot(\ell_s)$ is a semialternating pseudoline for $Rot(L_s)$. Similarly, $Rot(\ell_p)$ is semialternating for $Rot(L_p)$.
Recall that any M-semialternating pseudoline $\ell$ admitted by an arrangement $\mathcal{A}$ intersects every line of $\mathcal{A}$ once. Let $B_\ell$ denote the set of lines below the first cell in the underlying cell-path for $\ell$ in $\mathcal{A}$. If the level-sequence of the underlying cell-path for $Rot(\ell)$ in $Rot(\mathcal{A})$ becomes non-monotone, $B_\ell$ must contain both blue and red lines (note that $Rot(\ell)$ remains semialternating and cannot intersect any line in $Rot(\mathcal{A})$ more than once).
$B_{\ell_p}$ in $L_p$ is an empty set, and $B_{\ell_s}$ in $L_s$ consists of lines that are of the same color.
Thus, $Rot(\ell_s)$ and $Rot(\ell_p)$ are M-semialternating pseudolines for $Rot(L_s)$ and $Rot(L_p)$, respectively.
$Rot(\ell_s)$ starts and ends in cells of levels zero and $|L_s|$ in $Rot(L_s)$, and $Rot(\ell_p)$ starts and ends in median cells in $Rot(L_p)$. Therefore, $\mathcal{P}'$ contains a parallel set of size $ss(\mathcal{P})$, and a spoke set of size $ps(\mathcal{P})$.

In the following, we prove (in two ways) that if $\mathcal{P}$ has the additional property that $\mathcal{B}$ avoids $\mathcal{R}$ (i.e. $\mathcal{P}$ is $1$-avoiding), then $\mathcal{P}'$ is also $1$-avoiding.

Let ${\mathcal{P}}=\{p_1,p_2,\dots,p_n\}$, where $p_i=(a_i,b_i)$. We have
\begin{align*}
	\mathcal{L} &= \{p^\star_i=\{(x,y) : y=a_ix-b_i\} \mid i \in [n] \}, \\
	\mathcal{L}'&=\{ \{(x,y) : y=-\frac{x}{a_i}-\frac{b_i}{a_i}\} \mid i \in [n] \},  \text{~and} \\
	\mathcal{P}'&=\{ p'_i = (-\frac{1}{a_i},\frac{b_i}{a_i}) \mid i \in [n] \}.
\end{align*}
Note that since no points in $\mathcal{P}$ are on the $y$-axis, $\mathcal{L}'$ and $\mathcal{P}'$ are well-defined.
Let $o(p,q,r)$ denote the orientation of points $p$, $q$, and $r$. That is, $o(p,q,r)=+1$ if the circle through $p,q,r$ is clockwise; $o(p,q,r)=-1$ if the circle through $p,q,r$ is counterclockwise; and $o(p,q,r)=0$ if the three points are collinear.

Clearly, three points in ${\mathcal{P}}$ are collinear if and only if their corresponding points in $\mathcal{P}'$ are collinear.
Let $p_i$, $p_j$, and $p_k$ be three points in ${\mathcal{P}}$. Without loss of generality, we assume $p_i$ and $p_j$ are on the same side of the $y$-axis. Note that the ordering of the $x$-coordinates of $p_i$ and $p_j$ is the same as that of $p'_i$ and $p'_j$.
Suppose $p_k$ is above the line through $p_i$ and $p_j$. We show that $p'_k$ is above the line through $p'_i$ and $p'_j$ if and only if an even number of points in $\{p_i,p_j,p_k\}$ have negative $x$-coordinates. Similarly, if $p_k$ is below the line through $p_i$ and $p_j$, then $p'_k$ is below the line through $p'_i$ and $p'_j$ if and only if an even number of points in $\{p_i,p_j,p_k\}$ have negative $x$-coordinates. This is because

\begin{align*}
o(p_i,p_j,p_k) &= a_1b_2-a_2b_1+a_2b_3-a_3b_2+a_3b_1-a_1b_3 \text{, and}\\
o(p'_1,p'_j,p'_k)&=-\frac{b_2}{a_1a_2}+\frac{b_1}{a_1a_2}-\frac{b_3}{a_2a_3}+\frac{b_2}{a_2a_3}-\frac{b_1}{a_1a_3}+\frac{b_3}{a_1a_3},
\end{align*}
which implies $o(p_i,p_j,p_k)=a_1a_2a_3 \cdot o(p'_i,p'_j,p'_k)$. Therefore,
\begin{align*}
o(p_i,p_j,p_k) = 
\begin{cases}
o(p'_i,p'_j,p'_k)  &  \text{if } \text{neg}_x(p_i,p_j,p_k) \equiv 0 \pmod{2},\\
-o(p'_i,p'_j,p'_k)  &  \text{otherwise},
\end{cases}
\end{align*}
where $\text{neg}_x(p_i,p_j,p_k)$ is the number of points in $\{p_i,p_j,p_k\}$ that are on the left side of the $y$-axis. This immediately implies that if ${\mathcal{P}}={\mathcal{P}}_{{\mathcal{B}} \vdash {\mathcal{R}}}$ is $1$-avoiding and ${\mathcal{B}}$ and ${\mathcal{R}}$ are on different sides of the $y$-axis, then $\mathcal{P}'=\mathcal{B}' \cup \mathcal{R}'$ is a $1$-avoiding point set where $\mathcal{B}'$ avoids $\mathcal{R}'$, and $\mathcal{B}'$ and $\mathcal{R}'$ are on different sides of the $y$-axis.

Alternatively, we can prove $\mathcal{P}'$ is $1$-avoiding using the dual plane directly. Recall that ${\mathcal{P}}={\mathcal{P}}_{{\mathcal{B}} \vdash {\mathcal{R}}}$, and $\mathcal{L}={\mathcal{P}^\star}={\mathcal{B}^\star} \cup {\mathcal{R}^\star}$ is such that ($\uppercase\expandafter{\romannumeral 1}$) each line in ${\mathcal{R}}^\star$ intersects the lines of ${\mathcal{B}}^\star$ in the same order, and ($\uppercase\expandafter{\romannumeral 2}$) the lines above (or below) each unbounded median cell of $\mathcal{L}$ are all of the same color.
Note that both properties ($\uppercase\expandafter{\romannumeral 1}$) and ($\uppercase\expandafter{\romannumeral 2}$) hold for $Rot(\mathcal{L}')$, and hence the dual of $\mathcal{L}'$ is a $1$-avoiding point set. 
\end{proof}

A \emph{parallel family} is a set of segments such the supporting lines of every pair of segments intersect outside both segments.
The transformation we use for proving Lemma~\ref{lem:pset-ss} turns out to be the same, upto a negation in the $x$-coordinate, as the transformation used by \citet{Pach-crf} in proving that the problems of finding the maximum crossing family and maximum parallel family are equivalent.


\begin{conjecture}
Let $\mathcal{P}=\mathcal{P}_{\mathcal{B} \vdash \mathcal{R}}$ be a $1$-avoiding $2n$-point set. The sizes of the largest spoke set and largest parallel set for $\mathcal{P}$ are of the same order of magnitude.
\label{q:ssTOps}
\end{conjecture}

\begin{lemma}\label{lem:semiTOm}
Let $\mathcal{L}$ be a color-separable line arrangement. The size of the largest M-semialternating path in $\mathcal{L}$ is at least half the size of the largest semialternating path in $\mathcal{L}$.
\end{lemma}

\begin{proof}
Let $\ell$ be a semialternating pseudoline for subarrangement  $\mathcal{L}' \subseteq \mathcal{L}$.
Let $\mathcal{C}=\langle C_0, C_1,\cdots,C_{2k}\rangle$ be the underlying cell-path for $\ell$ in $\mathcal{L}'$.
Let $\mathcal{S}=\langle s_0,s_1,\cdots, s_k \rangle$ denote the level-sequence of $\mathcal{C}$. Clearly, for any $i \in [k]$, $s_i-s_{i-1} \in \{0,2,-2\}$. If for all $i \in [k]$, $s_i-s_{i-1} \in \{0,2\}$ then the level-sequence is monotone. Let $\mathcal{I}=\{ i \mid  s_i-s_{i-1}=-2\}$. Let $B(i)$ denote the set of lines that are below $C_{2i}$. Note that $\ell$ intersects each line of $\mathcal{L}'$ once. This implies that for $i \in \mathcal{I}$, $B(i-1) \setminus B(i) \subseteq B(0)$. Therefore, $|\mathcal{I}| \le \dfrac{s_0}{2}$. 
Let $\mathcal{L}''$ be the arrangement $\mathcal{L}'$ without $ \bigcup \limits_{i \in \mathcal{I}} B(i-1) \setminus B(i)$. The underlying cell-path for $\ell$ in $\mathcal{L}''$ is monotone. Recall that $s_0 \le k$. Hence $|\mathcal{L}''| \ge \dfrac{|\mathcal{L}'|}{2}$. This concludes the proof.
\end{proof}

\begin{conjecture}\label{conj:large-m}
The size of the largest semialternating path for any $1$-avoiding line arrangement $\mathcal{L}$ is linear (in the number of lines in $\mathcal{L}$). 
\end{conjecture}

Conjectures \ref{conj:large-m} and \ref{q:ssTOps} (if answered affirmatively), together with Lemmas \ref{lem:ssORps} and \ref{lem:semiTOm}, imply that any $1$-avoiding point set has a linear-sized spoke set.

\begin{conjecture}\label{conj:max-over-min}
Let $\mathcal{L}$ be a $1$-avoiding line arrangement.
Let $\ell$ be a directed pseudoline that intersects each line in $\mathcal{L}$ once.
For $i=0,1,\dots,n$, let $\mathcal{A}_i$ be the largest sub-arrangement that admits an M-semialternating $\ell$ that starts in level $i$ of sub-arrangement $\mathcal{A}_i$. Let $\mathcal{A}=\{|\mathcal{A}_i| \mid i \in [0..n],\mathcal{A}_i \neq \emptyset\}$.
The ratio of the maximum element in $\mathcal{A}$ to the minimum element in $\mathcal{A}$ is constant.
\end{conjecture}
Conjectures~\ref{q:ssTOps} and \ref{conj:max-over-min} are equivalent.

\begin{lemma}
Let $\mathcal{L}$ be a $1$-avoiding line arrangement consisting of $n$ blue lines and $n$ red lines. There exists $\mathcal{L}$ whose largest semialternating path is of size $\frac{10}{11}\cdot 2n$.
\end{lemma}
\begin{proof}[Proof sketch]
First, we construct a $1$-avoiding line arrangement consisting of $11$ blue lines and $11$ red lines such that it admits no semialternating pseudoline. See Figure~\ref{fig:noFullPath}. Each blue line has a greater slope than that of a red line. The blue lines are almost parallel. Let $b_i$ denote the $i$-th leftmost blue line. Let $r_i$ denote the red line with the $i$-th largest slope. For the sake of contradiction, suppose that there exists a semialternating pseudoline $l$. Note that the intersection of $l$ and $b_6$ may only appear on either the downward ray originating at $V$ or the upward ray originating at $U$ (otherwise, $l$ is not semialternating). We consider the following cases:
\begin{enumerate}[label={\sbf{Case \arabic*:}},leftmargin=2.7\parindent]
\item $l$ intersects the downward ray originating at $V$.
In order for $l$ to intersect each blue line in $\{b_1,\dots,b_5\}$, it needs to intersect $r_4$ prior to $b_6$. If the intersection of $l$ and $b_6$ is below $r_4$, then in order for $l$ to be semialternating, it needs to cross $r_4$ again, which is not possible. Thus, the intersection of $l$ and $b_6$ is above $r_4$ (and $l$ starts below $r_4$). Let $l_1,\dots, l_{22}$ denote the order in which $l$ intersects the lines of the arrangement. For an odd $i$, we say that $l_i$ and $l_{i+1}$ are paired. Recall that $l$ needs to intersect each blue line in $\{b_1,\dots,b_5\}$. However, there are at most four red lines that may be paired with $b_1,\dots b_5$. Hence, either $l$ cannot be semialternating or it needs to double cross some red lines.
\item $l$ intersects the upward ray originating at $U$.
In order for $l$ to intersect each blue line in $\{b_1,\dots,b_5\}$, it needs to intersect $r_{11}$ prior to $b_6$. This implies that, in order for $l$ to intersect every line, it needs to intersect red lines $\{r_5,\dots,r_{10}\}$ prior to $r_{11}$. Therefore, $l$ needs to intersect at least seven red lines prior to $b_6$, which is not possible.
\end{enumerate}
\begin{figure}[!ht]
\begin{center}
\resizebox{15cm}{15cm}{
\includegraphics{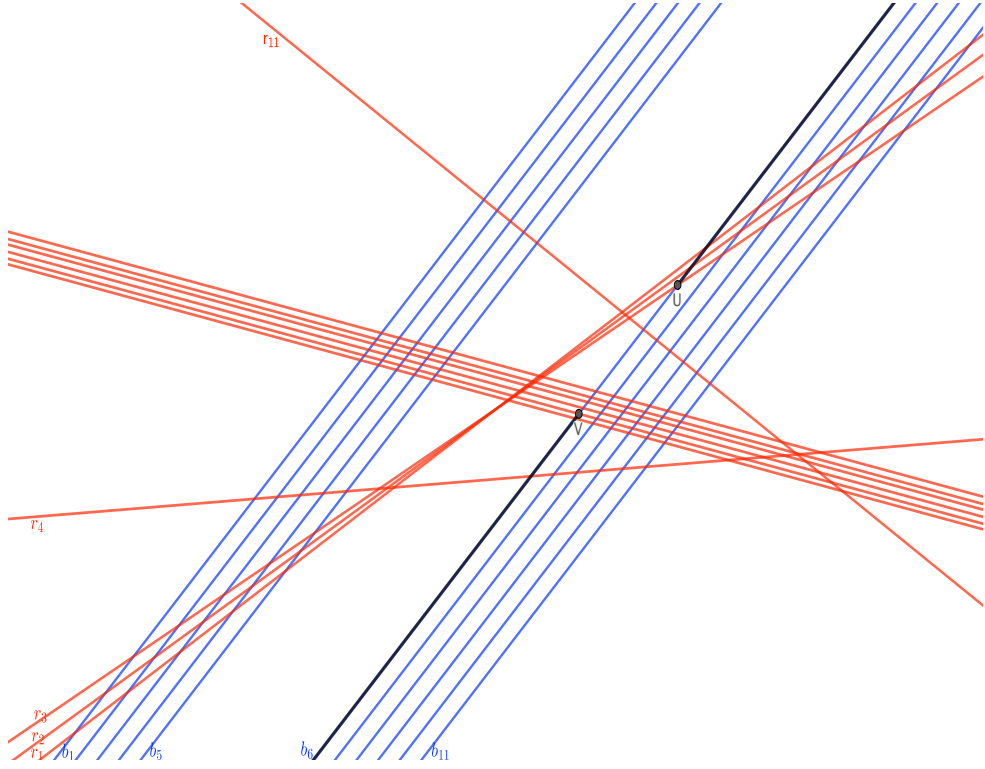}
}
\end{center}
\caption{A $1$-avoiding line arrangement that admits no semialternating pseudoline.}
\label{fig:noFullPath}
\end{figure}
By replacing each line of the arrangement depicted in Figure~\ref{fig:noFullPath} with $\frac{n}{11}$ lines of the same color such that they are all almost parallel to the initial line and their distances to each other are all infinitesimally small, we obtain the lemma. 
\end{proof}

\begin{lemma}
\label{lem:fps-c}
There exist $1$-avoiding point sets whose largest focal parallel set is of constant size.
\end{lemma}

\begin{proof}
We construct an arrangement of blue lines $\mathcal{L_{\mathcal{B}}}$ and red lines $\mathcal{L_{\mathcal{R}}}$ whose dual point sets $\mathcal{B}$ and $\mathcal{R}$ are mutually avoiding and the largest focal parallel set of $\mathcal{B} \cup \mathcal{R}$ is of constant size.
\begin{enumerate}[label={(S\arabic*)}]
\item Let $\{b_1,b_2,\dots, b_n\}$ be a set of blue vertical lines where $b_i$ is $x=i$.\label{b-unit}
\item Let $\{r_1,r_2,\dots,r_n\}$ be a set of red horizontal lines, where $r_1$ and $r_2$ are $y=1$ and $y=2$, respectively.
\item We define the remainder of the red lines incrementally. For $1< i < n$, we define $r_{i+1}$ from $\{r_1,r_i\}$ and $\{b_1,b_2,b_n\}$. Let $\mathtt{x}(l_1,l_2)$ denote the intersection point of lines $l_1$ and $l_2$. Let $d_i$ be the line joining $\mathtt{x}(r_1,b_1)$ and $\mathtt{x}(r_i,b_2)$. Let $x_i=\mathtt{x}(d_i,b_n)$. We define $r_{i+1}$ to be a horizontal line that is slightly above $x_i$. See Figure~\ref{fig:fps}.\label{construct-r}
\end{enumerate}
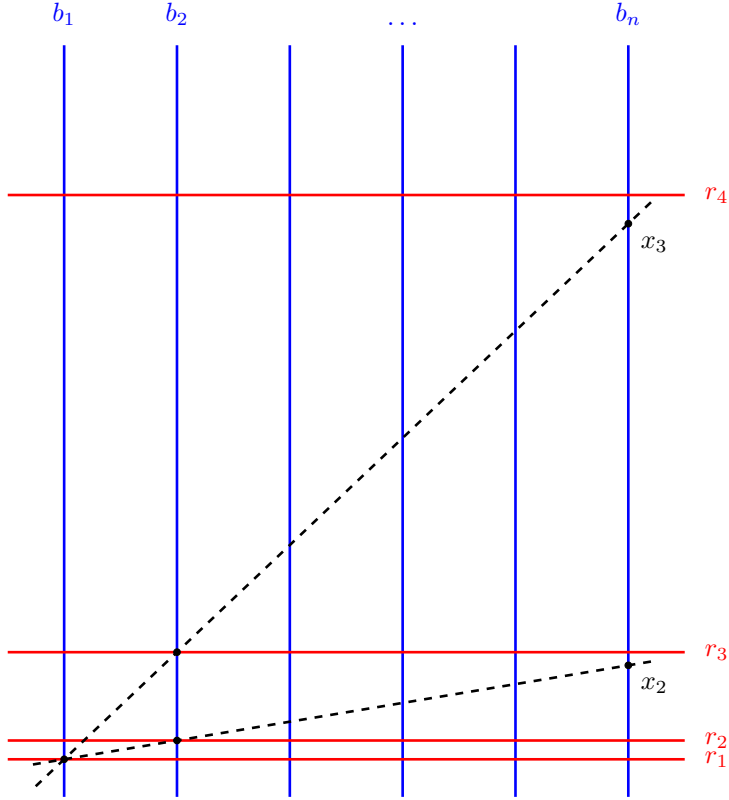
\begin{figure}[h]
\begin{center}
\begin{tikzpicture}[scale=.5,extended line/.style={shorten >=-#1,shorten <=-#1},
  extended line/.default=1cm]

	\def \lw {1};
	\def \d {3};
	\def \h {20}
	\def \n {6};

	\foreach \i in {1,2,...,\n}
		\draw[blue,line width=\lw pt,name path=B\i] (\i*\d-\d,0) --+ (0,\h) node [label=above:{\ifnum \i<3 $b_\i$ \else
		\ifnum \i=6 $b_n$ \else 
		\ifnum \i=4 $\ldots$
		\fi \fi \fi
		}] {} ;

	\path[draw,red,line width=\lw pt, name path=R1] (-\d/2,1) --+(\n*\d,0) node [label=right:{$r_1$}] {};
	\path[draw,red,line width=\lw pt, name path=R2] (-\d/2,1.5) --+(\n*\d,0) node [label=right:{$r_2$}] {};

	\path [name intersections={of=R1 and B1,by=X1}];
	\path [name intersections={of=R2 and B2,by=X2}];
	
	\draw [add = 5.2 and -.7, dashed,name path=D2,line width=1pt] (X1) to +($(X1)-(X2)$);
	\path [name intersections={of=D2 and B\n,by=X3}];
	\node[inner sep=1pt,label=below right:$x_2$] at (X3) {};

	\path[draw,red,line width=\lw pt, name path=R3] let \p1=(X3) in
	   (-\d/2,1.1*\y1) --+(6*\d,0) node [label=right:{$r_3$}] {};

	\path [name intersections={of=R3 and B2,by=X4}];
	\draw [add = 5.2 and -.7, dashed,name path=D3,line width=1pt] (X1) to +($(X1)-(X4)$);
	\path [name intersections={of=D3 and B\n,by=X5}];
	\node[inner sep=1pt,label=below right:$x_3$] at (X5) {};	
	
	\path[draw,red,line width=\lw pt, name path=R4] let \p1=(X5) in
	   (-\d/2,1.05*\y1) --+(6*\d,0) node [label=right:{$r_4$}] {};

	\foreach \i in {1,2,3,4,5}
		\draw[fill] (X\i) circle (2.5*\lw pt);

\end{tikzpicture}
\end{center}
\caption{Placement of red lines incrementally.}
\label{fig:fps}
\end{figure}

Each red line intersects the blue lines in the same order. Similarly, each blue line intersects the red lines in the same order. Hence $\mathcal{B} \cup \mathcal{R}$ is mutually avoiding. Let $\mathcal{L}=\mathcal{L_{\mathcal{R}}} \cup \mathcal{L_{\mathcal{B}}}$. To make sure that the largest focal parallel set of $\mathcal{B} \cup \mathcal{R}$ is of constant size, we need to show that no line added to $\mathcal{L}$ can be M-semialternating for a subset of $\mathcal{L}$ that is of size $\omega(1)$.

Note that any line $l$ intersects the blue lines monotonically (in order of their indices). Likewise, any line $l$ intersects the red lines monotonically. Let $s(l)$ denote the order in which the line $l$ intersects the lines in $\mathcal{L}$. If a line $l$ admits an M-semialternating path of size $m$, then there exists a subset $\mathcal{I} \subset [n]$ of size $\Theta(m)$ such that for every $i \in \mathcal{I}$, $r_i$ and $r_{i+1}$ do not appear consecutively in $s(l)$. For all $i \in \mathcal{I}$ except for at most one, the difference between the $x$-coordinates of $\mathtt{x}(l,r_i)$ and $\mathtt{x}(l,r_{i+1})$ is greater than or equal to one. Let $i$ be the smallest such index. The slope of the line through $\mathtt{x}(l,r_i)$ and $\mathtt{x}(l,r_{i+1})$ is less than the slope of $d_{i+1}$ (by our construction). Hence $b_n$ appears before $r_{i+2}$ in $s(l)$, which implies that $|\mathcal{I}|$ is constant.
 
In order to make sure that the dual point set $\mathcal{B} \cup \mathcal{R}$ is in general position, we perturb the blue and red lines imperceptibly so that all 
lines of the same colour become ``almost''  parallel.
A \emph{grid} point of $\mathcal{L}=\mathcal{L}_\mathcal{R} \cup \mathcal{L}_\mathcal{B}$ is an intersection point of a red line in $\mathcal{L}_\mathcal{R}$ and a blue line in $\mathcal{L}_\mathcal{B}$ .  Let $\mathcal{R}_i =\{r_1,r_2,\dots,r_i\}$ and let $\mathcal{L}_i=\mathcal{R}_i \cup \mathcal{L}_\mathcal{B}$. Note that the modification of the arrangement should not violate the property that every line through two grid points of $\mathcal{L}_i$, that is not a line in $\mathcal{L}_i$, intersects all blue lines prior to intersecting $r_{i+1}$. Lastly, we may rotate the resulting arrangement $45$ degrees clockwise so that the blue and red points (of the dual point configuration) are to the left and right sides of the $y$-axis.
\end{proof}

A point set $P$ is a \emph{wheel} if there exists a dummy point $q \notin P$ in the plane such that if we start with a vertical line through $q$ and rotate it clockwise about $q$ until it gets vertical again,
the rotating line would encounter points of $P$ on alternating sides of the vertical line through $q$.
A two-coloured point set is an \emph{alternating wheel} if it is a wheel where the rotating line through the dummy point sees the points with alternating colours. The size of a wheel is the number of points it has.

\begin{observation} It is easy to see the following:
\begin{iitem}
{\setlength\itemindent{20pt} \item A wheel of size $k$ implies a spoke set of size $\frac{k}{2}$.\label{obs:pin-sp}}
{\setlength\itemindent{20pt} \item The largest alternating wheel in a bicolored $2k$-point set that contains a crossing family of size $k$ consisting of bicolored segments (segments whose endpoints are not of the same color) may be of constant size.\label{obs:crf-pin}}
\end{iitem}
\end{observation}
\begin{proof}
Observation~\ref{obs:pin-sp} immediately follows from the definitions of wheels and spoke sets.
Note that an alternating wheel $P$ corresponds to an M-semialternating line in $P^\star$ that starts and ends in the median cells of $P^\star$. Lemma~\ref{lem:fps-c}, together with Lemma~\ref{lem:pset-ss} (when considering M-semialternating lines rather than pseudolines), proves Observation~\ref{obs:crf-pin}.
 \end{proof}

\begin{lemma}
There exists a $1$-avoiding $n$-point set whose largest parallel set is of size at most $\frac{n}{4}+1$.
\label{lem:p-set-upex}
\end{lemma}

\begin{proof}
We construct a two-coloured point set $\mathcal{P}=\mathcal{P}_{\mathcal{B} | \mathcal{R}}$. Let $\mathcal{B}$ and $\mathcal{R}$ be the set of blue and red points, respectively. We construct $\mathcal{B}$ such that the blue points are almost collinear, that is, they are on an arc of a very large circle. We construct $\mathcal{R}$ in the same way and make sure that every line joining two red points is a halving line for $\mathcal{B}$.
See Figure~\ref{fig:psUPex}.
Assume $\mathcal{B}$ and $\mathcal{R}$ are almost on vertical and horizontal lines, respectively, and let $\mathcal{B}$ be to the right of $\mathcal{R}$.
Let $\mathcal{R}=\mathcal{R}_1 \cup \mathcal{R}_2$, where $\mathcal{R}_1$ and $\mathcal{R}_2$ are the first and second halves of $\mathcal{R}$ when traversed from left to right. Let $\mathcal{B}=\mathcal{B}_1 \cup \mathcal{B}_2$, where $\mathcal{B}_1$ and $\mathcal{B}_2$ are the first and second halves of $\mathcal{B}$ when traversed from top to bottom. Let $L$ be a maximum parallel set exhausting $R \cup B$, where $R \subseteq \mathcal{R}$ and $B \subseteq \mathcal{B}$. For each $L_i \in L$, where $i>0$, let $r_i=A(R,L_i) \setminus A(R,L_{i-1})$ and $b_i=A(B,L_i) \setminus A(B,L_{i-1})$.
\vspace{.2cm}
\begin{figure}[h]
	\begin{center}
\begin{tikzpicture}[scale=1]
	\def \l {15};
	\def \a {1.5};
	
	\coordinate(O) at (0,0);
	\foreach \angle  in {\a,-\a}
		\foreach \i in {0,1,2,3}{
			\coordinate(A) at (180+\i * \angle+\angle:\l cm);
			\draw[fill,blue](A) circle (1.5pt);
		}
	\draw[line width= 10pt, blue, opacity=.2,cap=round] (180+4.5* \a:\l cm) arc (180+4.5* \a:180-4.5* \a:\l cm);
	\node[label=below:\textcolor{blue}{\large{$\mathbf{\mathcal{B}}$}}] at (180+5* \a:\l cm) {};
	\node[label=right:\textcolor{blue!50}{$\hspace{.25cm}\mathcal{B}_2$}] at (180+2.5* \a:\l cm) {};
	\node[label=right:\textcolor{blue!50}{$\hspace{.25cm}\mathcal{B}_1$}] at (180-2.5* \a:\l cm) {};

	\begin{scope}[shift={(-1.2*\l,+\l)},rotate=90]
		\foreach \angle  in {\a,-\a}
			\foreach \i in {0,1,2,3}{
				\coordinate(A) at (180+\i * \angle+\angle:\l cm);
				\draw[red,fill](A) circle (1.5pt);
			}
	\draw[line width= 10pt, red, opacity=.2,cap=round] (180+4.5* \a:\l cm) arc (180+4.5* \a:180-4.5* \a:\l cm);
	\node[label=left:\textcolor{red}{\large{$\bf{\mathcal{R}}$}}] at (180-5* \a:\l cm) {};
	\node[label=below:\textcolor{red!50}{$\vspace{.5cm}\mathcal{R}_2$}] at (180+2.5* \a:\l cm) {};
	\node[label=below:\textcolor{red!50}{$\vspace{.5cm}\mathcal{R}_1$}] at (180-2.5* \a:\l cm) {};
	\end{scope}
\end{tikzpicture}
	\end{center}
	\caption{Point set with maximum parallel set of size $\frac{n}{4}+1$.}
	\label{fig:psUPex}
\end{figure}
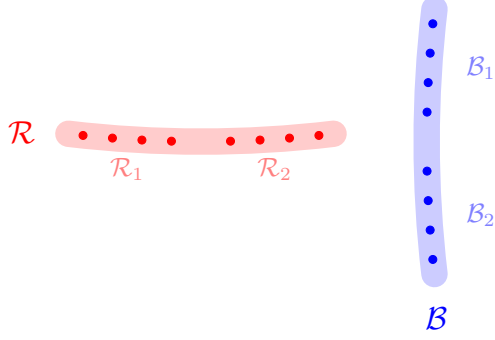

If $R \cap \mathcal{R}_1 = \emptyset$, then the parallel set is of size at most $\frac{n}{4}$. Assume $r_i \in R \cap \mathcal{R}_1$.
We consider two cases.
\begin{enumerate}[label={\sbf{Case \arabic*:}},leftmargin=2.7\parindent]
\item $b_i \in \mathcal{B}_1$. Note that for any $L_j \in L$, where $b_j \in \mathcal{B}_2$, $A(R,L_i) \subset A(R,L_j)$ (as $b_i$ is above $b_j$), and hence every red point in $R$ that is to the right of $r_i$ is in $ A(R,L_j)$. Therefore $|B \cap \mathcal{B}_2| \le 1$.
\item $b_i \in \mathcal{B}_2$. If $B \cap \mathcal{B}_1 = \emptyset$, then the parallel set is of size at most $\frac{n}{4}$. So we assume that there exists $L_j$ such that $b_j \in \mathcal{B}_1$. Since $b_j$ is above $b_i$, we know $j<i$. Thus $r_i \notin A(R,L_j)$, which implies $r_j$ is to the left of $r_i$.
Note that $A(R,L_{j})$ contains $r_j$ and all red point in $R$ to the left of $r_j$.
This implies that for any line $L_k$ such that $b_k \in \mathcal{B}_2$, $A(R,L_k)=R$. Therefore, $|B \cap \mathcal{B}_2|=1$.
\end{enumerate}
As a result, the size of the parallel set is at most $\max\{|\mathcal{B}_1|,|\mathcal{B}_2|\}+1 = \frac{n}{4}+1$.
\end{proof}

Note that Lemma~\ref{lem:p-set-upex} together with Lemma~\ref{lem:pset-ss} improves on the $\frac{9n}{20}$ upper bound known on the size of spoke sets due to \citet{sch}.

\begin{corollary}
There exists a $1$-avoiding $n$-point set whose largest spoke set is of size at most $\frac{n}{4}+1$.
\label{cor:ssUP}
\end{corollary}

Corollary~\ref{cor:ssUP} closes the gap between lower and upper bounds on the size of spoke sets for $1$-avoiding point sets if Conjecture~\ref{schnider-claim} is true.

The following observation compares the notions of crossing family, spoke set, parallel family and parallel set (for point sets in general position).

\begin{observation}
For a given point set $P$,
\begin{itemize}
\item if $P$ has a crossing family of size $k$, then there exists $P' \subseteq P$ where $|P'|=2k$ and $P'$ has $k$ pairwise crossing halving edges.
\item if $P$ has a spoke set of size $k$, then there exists $P' \subseteq P$ where $|P'|=2k$ and $P'$ has a set of $k$ halving lines $L$ such that the corresponding halving edge of each halving line is distinct and crosses all the halving lines in $L$.
\item if $P$ has a parallel family of size $k$, then there exists $P' \subseteq P$ where $|P'|=2k$ and $P'$ has a set of pairwise parallel $0-,2-,\dots,(2k-2)-$edges.
\item if $P$ has a parallel set of size $k$, then there exists $P' \subseteq P$ such that $|P'|=2k$ and $P'$ has a set of $0-,2-,\dots,2k-$sets, where for any $0\le i < k$, the $2i-$set is contained in the $2(i+1)-$set. A \emph{$k$-set} of point set $P$ is a subset $S$ of $P$ containing $k$ points that is separable from its complement $P \setminus S$ by a straight line.
\end{itemize}
\end{observation}

\subsection{Stabbing families} \label{sec:stab}
Here, we study a more generalized notion than spoke sets, for which we can easily prove a linear lower bound. We start by some definitions.

A \emph{stabbing family} is a set of segments such that for every pair of segments, the line extension of one intersects the interior of the other one.

\begin{definition}
Given two segments $e_1$ and $e_2$, extend the segments to obtain two lines. If the intersection of these lines lies
\begin{itemize}
\item  on both segments, we say the segments are \emph{crossing}.
\item  on $e_1$ but not on $e_2$, we say $e_2$ \emph{stabs} $e_1$.
\item  outside both segments, we say the segments are \emph{parallel}.
\end{itemize}
Two segments are \emph{non-crossing} if they are either parallel or one stabs the other.
\end{definition}

Note that any pair of segments in a stabbing family is either crossing or one stabs the other.
Let $\mathcal{L}_\mathcal{P}$ denote a spoke set for $\mathcal{P}$.
A \emph{spoke matching}, with respect to $\mathcal{L}_\mathcal{P}$, is a matching where each segment (i.e. matching edge) connects two points in $\mathcal{P}$ that lie in antipodal unbounded regions of the arrangement of $\mathcal{L}_\mathcal{P}$; and no two segments start (or end) in the same unbounded region.
It is easy to see that a spoke matching does not contain any parallel segments, and hence forms a stabbing family. Thus, any point set with a spoke set of size $k$, has a stabbing family of size at least $k$. However, the reverse is not true.~\citet{gcrf17} characterizes the family of spoke matchings and describes certain other properties that need to be satisfied by spoke matchings (See Theorem 2 in \citep{gcrf17}).


\begin{lemma}
Let $\mathcal{P}=\mathcal{P}_{\mathcal{B} \vdash \mathcal{R}}$ be a $1$-avoiding point set.
There is a perfect matching in $\mathcal{P}$ such that every edge of the matching connects a point in $\mathcal{R}$ to a point in $\mathcal{B}$ and the matching edges are all pairwise non-crossing.
\label{lem:non-cross}
\end{lemma}
\begin{proof}
Without loss of generality, assume $\mathcal{R}$ and $\mathcal{B}$ lie to the left and right of the $y$-axis, respectively. Label the points of $\mathcal{B}$ by $b_1b_2\dots b_n$ so that for all $i<j$, $\mathcal{R}$ lies to the right of the directed line from $b_i$ to $b_j$ (since $\mathcal{P}_{\mathcal{B} \vdash \mathcal{R}}$ is $1$-avoiding such a labeling exists). Start with $k=1$. For each $k$, rotate a vertical line through $b_k$, counterclockwise about $b_k$, until it hits an unmatched red point, say $r_k$. Match $b_k$ with $r_k$. Increment $k$ by one and repeat the last step as long as $k \le n$.
It is easy to see that the matching obtained has the desired property. For the sake of contradiction suppose $b_ir_i$ and $b_jr_j$ cross, where $i<j$.
Note that since the matching is constructed incrementally, $r_j$ is unmatched when $b_i$ and $r_i$ are matched. However, since both $r_i$ and $r_j$ are to the right of the directed line from $b_i$ to $b_j$ (and the assumption that $\mathcal{B}$ and $\mathcal{R}$ are to the left and right sides of the $y$-axis), we know that if we rotate a vertical line through $b_i$ counterclockwise about $b_i$, it sees $r_j$ prior to $r_i$. This implies that the algorithm picks $r_j$ (over $r_i$) for $b_i$, and hence $b_ir_i$ cannot be a matching edge.
\end{proof}

\begin{corollary}
Let $\mathcal{P}=\mathcal{P}_{\mathcal{B} \vdash \mathcal{R}}$ be a $1$-avoiding point set.
There is a perfect matching in $\mathcal{P}$ such that every edge of the matching connects a point in $\mathcal{R}$ to a point in $\mathcal{B}$ and the matching edges form a stabbing family.
\end{corollary}

\begin{proof}
Assume (by rotation and transformation if necessary) that $\mathcal{B}$ and $\mathcal{R}$ lie on the right and left sides of the $y$-axis, respectively.
We transform ${\mathcal{P}}$ to a new point set $\mathcal{P}'$ so that
the dual line arrangement for $\mathcal{P}'$ is a rotation of the dual line arrangement for $\mathcal{P}$ by $90\degree$ (similar to the transformation used in Lemma~\ref{lem:pset-ss}). Recall that a segment $s$ in the primal plane transforms to a double wedge $W_s$ in the dual plain. We refer to the point representing the dual of the supporting line of $s$ as the \emph{apex} of $W_s$. Two segments in the primal plane cross if and only if in the dual plane, the apex of each double wedge is inside the other double wedge. Two segments in the primal plane are parallel if and only if in the dual plane, the apex of neither double wedge is inside the other double wedge. Segment $e$ stabs segment $f$ if and only if in the dual plane, the apex of $W_e$ is inside $W_f$, and the apex of $W_f$ is outside $W_e$.
Let the segment $e_i$ within ${\mathcal{P}}$ transform to the segment $e'_i$ within $\mathcal{P}'$. Note that for a bicolored segment $e_i$, $W_{e_i}$ contains a horizontal line. Hence, the complementary double wedge of $W_{e_i}$ when rotated $90\degree$ is the dual of $e'_i$. Therefore, for a pair of bicolored segments $e_i$ and $e_j$,
\begin{itemize}
\item $e_i$ and $e_j$ are crossing if and only if $e'_i$ and $e'_j$ are parallel,
\item $e_i$ stabs $e_j$ if and only if $e'_j$ stabs $e'_i$, and
\item $e_i$ and $e_j$ are parallel if and only if $e'_i$ and $e'_j$ are crossing.
\end{itemize}
We prove in Lemma~\ref{lem:pset-ss} that if $\mathcal{P}$ is $1$-avoiding, so is $\mathcal{P}'$. Lemma~\ref{lem:non-cross} implies that $\mathcal{P}'$ has a perfect bicolored matching whose segments are pairwise non-crossing. Therefore, ${\mathcal{P}}$ has a perfect bicolored matching that forms a stabbing family.
\end{proof}

We can easily generalize this result to general point sets.
\begin{lemma}
The largest stabbing family for any $2n$-point set $\mathcal{P}$ in general position is of size $n$.
\end{lemma}

\begin{proof}
Translate $\mathcal{P}$ so that the $y$-axis becomes a halving line in $\mathcal{P}$. Assume the right and left halves in $\mathcal{P}$ are blue and red respectively. 
Similar to what we did before,
we transform ${\mathcal{P}}$ to a new point set $\mathcal{P}'$ so that the dual line arrangement for $\mathcal{P}'$ is a rotation of the dual line arrangement for $\mathcal{P}$ by $90\degree$.

It is a well-known result that every two-colored point set admits a line, called a ham-sandwich cut, that simultaneously bisects each color class.
We say that a matching is \emph{non-crossing} if the matching edges are pairwise non-crossing.
We find a non-crossing bicolored matching in $\mathcal{P}'$ by induction on $|\mathcal{P}'|$.
If $|\mathcal{P}'|=2$, we match the two points. Otherwise, we find a ham-sandwich cut splitting $\mathcal{P}'$ into two subsets each containing half the red and half the blue points. We match the points lying on a ham-sandwich cut (if any) and find non-crossing bicolored matchings in each of the two (smaller) subsets.
A non-crossing bicolored matching in $\mathcal{P}'$ corresponds to a bicolored matching forming a stabbing family in $\mathcal{P}$.
\end{proof}

\paragraph{\bf{Acknowledgments.}} An initial upper bound of $\frac{n}{4}$ on the size of crossing families was achieved during the second author's visit at EPFL. We thank G\'abor Tardos and J\'anos Pach for helpful discussion at EPFL.

\bibliography{references}

\begin{thebibliography}{22}
\providecommand{\natexlab}[1]{#1}
\providecommand{\url}[1]{\texttt{#1}}
\expandafter\ifx\csname urlstyle\endcsname\relax
  \providecommand{\doi}[1]{doi: #1}\else
  \providecommand{\doi}{doi: \begingroup \urlstyle{rm}\Url}\fi

\bibitem[Ackerman(2009)]{quasi-4}
Eyal Ackerman.
\newblock On the maximum number of edges in topological graphs with no four
  pairwise crossing edges.
\newblock \emph{Discrete and Computational Geometry}, 41\penalty0 (3):\penalty0
  365--375, April 2009.

\bibitem[Ackerman and Tardos(2007)]{quasi-planar}
Eyal Ackerman and Gábor Tardos.
\newblock The maximum number of edges in quasi-planar graphs.
\newblock \emph{Journal of Combinatorial Theory, Series A}, 114:\penalty0
  563–571, April 2007.

\bibitem[Aichholzer and Krasser(2001)]{g-k10}
Oswin Aichholzer and Hannes Krasser.
\newblock The point set order type data base: A collection of applications and
  results.
\newblock In \emph{Proceedings of the Thirteen Canadian Conference on
  Computational Geometry}, pages 17--20, January 2001.

\bibitem[Aichholzer et~al.(2014)Aichholzer, Cardinal, Hackl, Hurtado, Korman,
  Pilz, I.~Silveira, Uehara, Valtr, Vogtenhuber, and Welzl]{cell-path-Aich}
Oswin Aichholzer, Jean Cardinal, Thomas Hackl, Ferran Hurtado, Matias Korman,
  Alexander Pilz, Rodrigo I.~Silveira, Ryuhei Uehara, Pavel Valtr, Birgit
  Vogtenhuber, and Emo Welzl.
\newblock Cell-paths in mono- and bichromatic line arrangements in the plane.
\newblock \emph{Discrete Mathematics and Theoretical Computer Science}, 16,
  January 2014.

\bibitem[Alvarez-Rebollar et~al.(2015)Alvarez-Rebollar, Cravioto-Lagos, and
  Urrutia]{cr-ham}
Jose~Luis Alvarez-Rebollar, Jorge Cravioto-Lagos, and Jorge Urrutia.
\newblock Crossing families and self crossing hamiltonian cycles.
\newblock In \emph{Abstracts of the XVI Spanish Meeting on Computational
  Geometry}, pages 13--16, Barcelona, July 2015.

\bibitem[Angelini et~al.(2017)Angelini, Bekos, Brandenburg, Da~Lozzo,
  Di~Battista, Didimo, Liotta, Montecchiani, and Rutter]{k-quasi-rel}
Patrizio Angelini, Michael Bekos, Franz Brandenburg, Giordano Da~Lozzo,
  Giuseppe Di~Battista, Walter Didimo, Giuseppe Liotta, Fabrizio Montecchiani,
  and Ignaz Rutter.
\newblock On the relationship between $k$-planar and $k$-quasi planar graphs.
\newblock In Hans~L. Bodlaender and Gerhard~J. Woeginger, editors,
  \emph{Graph-Theoretic Concepts in Computer Science, WG 2017}, volume 10520 of
  \emph{Lecture Notes in Computer Science}, pages 59--74. 2017.

\bibitem[Aronov et~al.(1994)Aronov, Erd{\H{o}}s, Goddard, Kleitman, Klugerman,
  Pach, and Schulman]{Pach-crf}
B.~Aronov, P.~Erd{\H{o}}s, W.~Goddard, D.~J. Kleitman, M.~Klugerman, J.~Pach,
  and L.~J. Schulman.
\newblock Crossing families.
\newblock \emph{Combinatorica}, 14\penalty0 (2):\penalty0 127--134, 1994.

\bibitem[Bose et~al.(2005)Bose, Hurtado, Rivera-Campo, and Wood]{bose}
Prosenjit Bose, Ferran Hurtado, Eduardo Rivera-Campo, and David~R. Wood.
\newblock Partitions of complete geometric graphs into plane trees.
\newblock In J{\'a}nos Pach, editor, \emph{Graph Drawing}, pages 71--81.
  Springer Berlin Heidelberg, 2005.

\bibitem[Brandenburg(2016)]{k10}
Franz~J Brandenburg.
\newblock A simple quasi-planar drawing of ${K}_{10}$.
\newblock In Yifan Hu and Martin N{\"o}llenburg, editors, \emph{Graph Drawing
  and Network Visualization: 24th International Symposium, GD 2016}, volume
  9801 of \emph{Lecture Notes in Computer Science}, pages 603--604. 2016.

\bibitem[Dolores and Rubio-Montiel(2018)]{ghc18}
Lara Dolores and Christian Rubio-Montiel.
\newblock On crossing families of complete geometric graphs.
\newblock \emph{arXiv:1805.09888v2}, August 2018.

\bibitem[Fulek et~al.(2018)Fulek, Gärtner, Kupavskii, Valtr, and
  Wagner]{tverberg}
Radoslav Fulek, Bernd Gärtner, Andrey Kupavskii, Pavel Valtr, and Uli Wagner.
\newblock The crossing tverberg theorem.
\newblock \emph{arXiv:1812.04911v1}, December 2018.

\bibitem[Hoffmann et~al.(2015)Hoffmann, Kleist, and Miltzow]{cell-path-Linda}
Udo Hoffmann, Linda Kleist, and Tillmann Miltzow.
\newblock Upper and lower bounds on long dual paths in line arrangements.
\newblock In Giuseppe~F. Italiano, Giovanni Pighizzini, and Donald~T. Sannella,
  editors, \emph{Mathematical Foundations of Computer Science 2015}, volume
  9235 of \emph{Lecture Notes in Computer Science}, pages 407--419. Springer
  Berlin Heidelberg, 2015.

\bibitem[Kaneko and Kano(2003)]{rb-points}
Atsushi Kaneko and M.~Kano.
\newblock Discrete geometry on red and blue points in the plane --- a survey
  ---.
\newblock In Boris Aronov, Saugata Basu, J{\'a}nos Pach, and Micha Sharir,
  editors, \emph{Discrete and Computational Geometry: The Goodman-Pollack
  Festschrift}, volume~25 of \emph{Algorithms and Combinatorics}, pages
  551--570. Springer Berlin Heidelberg, 2003.

\bibitem[Pach and Solymosi(1999)]{sol-hl}
J\'anos Pach and Jozsef Solymosi.
\newblock Halving lines and perfect cross-matchings.
\newblock In \emph{Advances in discrete and computational geometry}, volume 223
  of \emph{Contemporary Mathematics}, pages 245--249. American Mathematical
  Society, Providence, Rhode Island, 1999.

\bibitem[Pach et~al.(2019)Pach, Rubin, and Tardos]{pach-new}
J\'{a}nos Pach, Natan Rubin, and G\'abor Tardos.
\newblock Planar point sets determine many pairwise crossing segments.
\newblock \emph{arXiv:1904.08845v1}, April 2019.

\bibitem[Schnider(2015)]{sch}
Patrick Schnider.
\newblock Partitions and packings of complete geometric graphs with plane
  spanning double stars and paths.
\newblock Master's thesis, ETH Z\"urich, Switzerland, 2015.

\bibitem[Schnider(2017)]{gcrf17}
Patrick Schnider.
\newblock A generalization of crossing families.
\newblock \emph{arXiv:1702.07555}, February 2017.

\bibitem[Suk and Walczak(2015)]{suk-k-quasi}
Andrew Suk and Bartosz Walczak.
\newblock New bounds on the maximum number of edges in k-quasi-planar graphs.
\newblock \emph{Computational Geometry: Theory and Applications}, 50\penalty0
  (C):\penalty0 24--33, December 2015.

\bibitem[Valtr(1996)]{Valtr-dense}
Pavel Valtr.
\newblock Lines, line-point incidences and crossing families in dense sets.
\newblock \emph{Combinatorica}, 16\penalty0 (2):\penalty0 269--294, 1996.

\bibitem[Valtr(1997{\natexlab{a}})]{Valtr}
Pavel Valtr.
\newblock On mutually avoiding sets.
\newblock In Ronald~L. Graham and Jaroslav Ne{\v{s}}et{\v{r}}il, editors,
  \emph{The Mathematics of Paul Erd{\"o}s II}, pages 324--328. Springer Berlin
  Heidelberg, 1997{\natexlab{a}}.

\bibitem[Valtr(1997{\natexlab{b}})]{valtr-kc}
Pavel Valtr.
\newblock Graph drawings with no k pairwise crossing edges.
\newblock \emph{Lecture Notes in Computer Science}, 1353, July
  1997{\natexlab{b}}.

\bibitem[Valtr(1998)]{valtr-kp}
Pavel Valtr.
\newblock On geometric graphs with no k pairwise parallel edges.
\newblock \emph{Discrete {\&} Computational Geometry}, 19\penalty0
  (3):\penalty0 461--469, March 1998.

\end{thebibliography}

\end{document}